\documentclass[pdflatex,sn-mathphys-ay]{sn-jnl}

\usepackage{graphicx}%
\usepackage{multirow}%
\usepackage{amsmath,amssymb,amsfonts}%
\usepackage{amsthm}%
\usepackage{mathrsfs}%
\usepackage[title]{appendix}%
\usepackage{xcolor}%
\usepackage{textcomp}%
\usepackage{manyfoot}%
\usepackage{booktabs}%
\usepackage{algorithm}%
\usepackage{algorithmicx}%
\usepackage{algpseudocode}%
\usepackage{listings}%

\usepackage{setspace}
\usepackage{epsfig}
\usepackage{subfigure}
\usepackage{pifont,bm}
\usepackage{diagbox}
\usepackage{makecell}

\theoremstyle{thmstyleone}%
\newtheorem{proposition}{Proposition}[section]%
\newtheorem{remark}{Remark}[section]%

\begin{document}

\title[Article Title]{Numerical inverse scattering transform for the coupled modified Korteweg-de Vries equation}


\author[1]{\fnm{Wen-Xin} \sur{Zhang}}

\author*[1,2]{\fnm{Yong} \sur{Chen}}\email{ychen@sei.ecnu.edu.cn}

\affil[1]{\orgdiv{School of Mathematical Sciences, Key Laboratory of Mathematics and Engineering Applications (Ministry of Education) and Shanghai Key Laboratory of PMMP}, \orgname{East China Normal University}, \orgaddress{\city{Shanghai}, \postcode{200241}, \country{China}}}

\affil[2]{\orgdiv{College of Mathematics and Systems Science}, \orgname{Shandong University of Science and Technology}, \orgaddress{\city{Qingdao}, \postcode{266590}, \country{China}}}


\abstract{This paper develops the numerical inverse scattering transform (NIST) framework for the coupled modified Korteweg-de Vries (mKdV) equation based on its associated Riemann-Hilbert problem. The coupled system gives rise to a $3\times3$ matrix-valued Riemann-Hilbert problem, whose jump matrix and scattering data have a more involved structure than in the scalar case. This matrix setting makes the extension of NIST to the coupled system nontrivial, both in the direct scattering computation and in the numerical solution of the inverse problem. Within this framework, the scattering data are first computed by solving the matrix direct scattering problem using a Chebyshev collocation method with suitable mappings. The Deift-Zhou nonlinear steepest descent method is then used to analyze and deform the oscillatory Riemann-Hilbert problem. In particular, the phase function admits two stationary points symmetric about the origin, and the analysis leads to a division of the $(x,t)$-plane into three regions with corresponding contour deformations. Compared with traditional numerical methods, the NIST computes the solution directly at prescribed spatial and temporal points without relying on time-stepping. Numerical experiments illustrate the performance of the proposed NIST in long-time simulations and indicate that it captures the main asymptotic features of the coupled mKdV solutions.}




\maketitle

\section{Introduction}\label{sec1}

Among nonlinear evolution equations, integrable models have attracted sustained attention due to their rich mathematical structures and important physical applications. The modified Korteweg-de Vries (mKdV) equation \citep{Miura1968,Wadati1973} is an important completely integrable nonlinear equation with the following form
\begin{equation}\label{mKdV}
  u_t+u_{xxx}-6\alpha u^2u_x=0,
\end{equation}
which describes a variety of complex phenomena in mathematical physics, such as internal ocean waves \citep{Grimshaw1997}, ultra-short pulses in nonlinear optics \citep{Mel’nikov1997, Leblond2018}, Alfv\'{e}n waves \citep{Khater1998}, anharmonic lattices \citep{Ono1992}, fluid mechanics \citep{Helal2002}, and so on. In addition, the generalization of the mKdV equation to multi-component or matrix systems has been considered in many works \citep{Yajima1975,Athorne1987,Hirota1997}, and one typical extension is the coupled mKdV equation. \cite{Iwao1997} proposed the coupled mKdV equation and constructed its multi-soliton solutions via the Hirota bilinear method. Moreover, the coupled mKdV equation is studied by means of the inverse scattering method \citep{Tsuchida1998,Wu2017,Ma2019}, the Darboux transformation method \citep{Xue2015,Hu2009}, the Fokas method \citep{Tian2018}, the Riemann-Hilbert approach \citep{Ma2018,Xu2023}, and the tanh-function method \citep{Fan2001}. It is also shown that the coupled mKdV equation possesses infinitely many conservation laws \citep{Tsuchida1998,Xue2015}, while its algebro-geometric structure \citep{Geng2014} and the nonlinear stability of its $N$-soliton solutions \citep{Ling2025} are also analyzed.

The inverse scattering transform \citep{Gardner1967,Ablowitz1974,Ablowitz1981} plays an important role in the analysis of integrable nonlinear evolution equations, especially in the study of their initial value problems and soliton dynamics. Within the inverse scattering transform framework, the inverse problem can be reformulated as a Riemann-Hilbert problem \citep{Novikov1984,Yang2010}, thereby providing a more direct and effective approach for studying solutions and their long-time asymptotic behavior in integrable systems. For long-time asymptotic analysis, the Deift-Zhou nonlinear steepest descent method \citep{Deift1993,Deift1994} provides a powerful tool for oscillatory Riemann-Hilbert problems. Through suitable contour deformation, the original Riemann-Hilbert problem is transformed into a form suitable for asymptotic analysis. The Deift-Zhou nonlinear steepest descent method has been successfully applied to a variety of integrable equations, including the mKdV equation \citep{Deift1993}, the KdV equation \citep{Grunert2009}, the nonlinear Schr\"{o}dinger (NLS) equation \citep{Kamvissis1996,Buckingham2007}, the derivative NLS equation \citep{Xu2013}, the sine-Gordon equation \citep{Cheng1999}, the Camassa-Holm equation \citep{Monvel2009}, and the Kundu-Eckhaus equation \citep{Wang2018}. In particular, the long-time asymptotics for the Cauchy problem of the coupled mKdV equation are derived in \cite{Geng2019}, and those for the three-component coupled mKdV system are also investigated in \cite{Ma2019-2}.

Alongside these analytical developments, the numerical computation of solutions and their long-time behavior for integrable equations has also attracted considerable attention. Traditional numerical methods are typically implemented through time-stepping, making it difficult to resolve the long-time asymptotic behavior accurately. The numerical inverse scattering transform (NIST), introduced by \cite{Trogdon2012}, provides an effective framework for computing solutions to integrable equations, especially in long-time regimes. A notable advantage of the NIST is that it computes the solution directly at any given $(x,t)$ without relying on spatial discretization or time-stepping procedures \citep{Olver2011,Olver2012,Trogdon2016}. This framework has been successfully applied to the Cauchy initial-value problem for the KdV and mKdV equations, capturing dispersive wave structures and asymptotic solution behavior \citep{Trogdon2012,Trogdon2013,Bilman2020}. Beyond the KdV and mKdV equations, the NIST has also been extended to a range of other integrable systems, including the focusing and defocusing NLS equations \citep{Trogdon2013-2,Gkogkou2026}, the derivative NLS equation \citep{Cui2024}, the coupled NLS equation \citep{Zhang2025}, the Kundu-Eckhaus equation \citep{Cui2023}, the sine-Gordon equation \citep{Deconinck2019}, and the Toda lattice \citep{Bilman2017}.

In this paper, we consider the coupled mKdV equation \citep{Iwao1997} with initial data in the Schwartz space $S(\mathbb{R})=\left\{ f\in\mathbb{C}^{\infty}(\mathbb{R}),\|f\|_{\alpha,\beta}=\sup\left|x^{\alpha}\partial^{\beta}f(x)\right|<\infty,\ \alpha,\beta\in\mathbb{Z}_+\right\}$, which is formulated as
\begin{gather}\label{cmKdV}
\left\{\begin{array}{l}
u_t+u_{xxx}-3\alpha(u_xv^2+uvv_x+2u^2u_x)=0,\\
v_t+v_{xxx}-3\alpha(u^2v_x+uu_xv+2v^2v_x)=0,\\
u(x,0)=u_0(x),\ v(x,0)=v_0(x),
\end{array}\right.
\end{gather}
where the subscripts $t$ and $x$ denote the partial derivatives of the real-valued functions $u=u(x,t)$ and $v=v(x,t)$. By introducing the complex-valued function $q=u+i v$, the above coupled system can be rewritten as
\begin{equation}
  q_t+q_{xxx}-3\alpha|q|^2q_x-\frac{3\alpha}{2}q(|q|^2)_x=0.\notag
\end{equation}
The parameter $\alpha=\pm1$ determines the sign of the nonlinear interaction. In particular, by setting $v\equiv0$, the coupled system (\ref{cmKdV}) reduces to the scalar mKdV equation (\ref{mKdV}). Under the usual sign convention for (\ref{mKdV}), $\alpha=-1$ and $\alpha=1$ correspond to the focusing and defocusing cases, respectively.

To study the coupled mKdV equation (\ref{cmKdV}), we develop the NIST framework adapted to its $3\times3$ matrix spectral problem. This matrix spectral setting gives rise to a $3\times3$ scattering matrix and a Riemann-Hilbert problem with a $3\times3$ jump matrix. As a result, the computation of the scattering data, the factorization and deformation of the jump matrix, and the numerical solution of the inverse problem become substantially more involved. Starting from the Lax pair, we derive the symmetry properties of the scattering matrix and construct the associated Riemann-Hilbert problem through spectral analysis. In the direct scattering problem, the scattering data, including the vector-valued reflection coefficients and eigenvalues, are computed numerically with high accuracy using the Chebyshev collocation method together with appropriate mapping functions. The stability and convergence of this numerical direct scattering scheme are also analyzed. To handle the oscillatory Riemann-Hilbert problem, suitable factorizations of the jump matrix and contour deformations are introduced by means of the Deift-Zhou nonlinear steepest descent method. In particular, the analysis leads to a division of the $(x,t)$-plane into three regions, for which different contour deformation strategies are constructed. The resulting deformed Riemann-Hilbert problems are then solved numerically within the proposed NIST framework, revealing nontrivial wave patterns of the coupled mKdV equation. Numerical comparisons with traditional time-stepping methods illustrate the performance of the proposed NIST in computing long-time evolution and capturing the essential asymptotic structure of the solutions.

This paper is organized as follows. In Section \ref{sec2}, for the coupled mKdV equation, the spectral analysis is carried out and the associated Riemann-Hilbert problem is formulated. In Section \ref{sec3}, the numerical direct scattering is performed to compute the scattering data with high accuracy. In Section \ref{sec4}, the oscillatory Riemann-Hilbert problem is analyzed and deformed so that the oscillatory terms become exponentially decaying contributions. In Section \ref{sec5}, the NIST is implemented to solve the Riemann-Hilbert problem numerically, especially in the long-time regime. Finally, conclusions are given in Section \ref{sec6}.

\section{Construction of the Riemann-Hilbert problem}\label{sec2}

In this section, the Riemann-Hilbert problem associated with the coupled mKdV equation is formulated based on the spectral analysis of its Lax pair and the corresponding jump condition. We start the study from the Lax pair \citep{Tsuchida1998} of the coupled mKdV equation which is completely integrable
\begin{gather}\label{Lax}
\left\{\begin{array}{l}
Y_x-i \alpha k\Lambda Y=QY,\\
Y_t-4i \alpha k^3\Lambda Y=\widetilde{Q}Y,
\end{array}\right.
\end{gather}
where $Y=Y(x,t;k)$ is the eigenfunction concerning the spectral parameter $k$, and
\begin{gather}
\Lambda=\operatorname{diag}(-1,1,1),\notag\\
Q=\left(\begin{array}{ccc}
0 & \alpha u & \alpha v \\
u & 0 & 0\\
v & 0 & 0\\
\end{array}\right),\notag\\
\widetilde{Q}=4k^2Q+2i \alpha k \Lambda(Q^2-Q_x)-Q_{xx}+Q_xQ-QQ_x+2Q^3.\notag
\end{gather}
Then, the asymptotic property of the eigenfunction $Y(x,t;k)$ is given as
\begin{equation}
Y(x,t;k)\rightarrow e^{i \alpha k\Lambda x+4i \alpha k^3\Lambda t}, \quad |x|\rightarrow\infty.\notag
\end{equation}
We consider taking the transformation
\begin{equation}\label{YJtransform}
Y(x,t;k)=J(x,t;k)e^{i \alpha k\Lambda x+4i \alpha k^3\Lambda t},
\end{equation}
where the new matrix function $J(x,t;k)$ satisfies the asymptotic property $J(x,t;k)\rightarrow I$ when $x$ approaches infinity. Combining the transformation (\ref{YJtransform}) with the Lax pair (\ref{Lax}), $J(x,t;k)$ admits the following spatial and temporal matrix spectral problems
\begin{gather}\label{JLax}
\left\{\begin{array}{l}
J_x-i \alpha k[\Lambda,J]=QJ,\\
J_t-4i \alpha k^3[\Lambda,J]=\widetilde{Q}J,
\end{array}\right.
\end{gather}
where $[\Lambda,J]=\Lambda J-J\Lambda$. If $J_{1,2}(x,t;k)$ are the matrix Jost solutions of the above partial differential system (\ref{JLax}), then they can be uniquely determined by the following Volterra integral equations
\begin{gather}
J_1(x,t;k)=I+\int_{-\infty}^xe^{i\alpha k(x-s)\Lambda} Q(s)J_1(s,t;k)e^{-i\alpha k(x-s)\Lambda}ds,\notag\\
J_2(x,t;k)=I-\int_x^{+\infty}e^{i\alpha k(x-s)\Lambda} Q(s)J_2(s,t;k)e^{-i\alpha k(x-s)\Lambda}ds.\notag\
\end{gather}

\begin{proposition}[Analyticity]
Introducing the definitions $J_1=(J_{11},J_{12},J_{13})$ and $J_{2}=(J_{21},J_{22},J_{23})$, these three-dimensional column vectors in matrix Jost solutions $J_1(x,t;k)$ and $J_2(x,t;k)$ possess specific analytical properties
\begin{itemize}
\item $\alpha=1$: the $J_{11}$, $J_{22}$ and $J_{23}$ are bounded and analytic in the upper half-plane $\mathbb{C}_+=\{k\in \mathbb{C}| \operatorname{Im}(k)>0\}$, while the $J_{12}$, $J_{13}$ and $J_{21}$ are bounded and analytic in the lower half-plane $\mathbb{C}_-=\{k\in \mathbb{C}| \operatorname{Im}(k)<0\}$;
\item $\alpha=-1$: the $J_{11}$, $J_{22}$ and $J_{23}$ are bounded and analytic in the lower half-plane $\mathbb{C}_-=\{k\in \mathbb{C}| \operatorname{Im}(k)<0\}$, while the $J_{12}$, $J_{13}$ and $J_{21}$ are bounded and analytic in the upper half-plane $\mathbb{C}_+=\{k\in \mathbb{C}| \operatorname{Im}(k)>0\}$.
\end{itemize}
\end{proposition}

Since both the matrix Jost solutions $J_1(x,t;k)$ and $J_2(x,t;k)$ satisfy the spatial and temporal matrix spectral problems (\ref{JLax}), then they have the following linear relationship
\begin{equation}\label{S}
J_{1}(x,t;k)=J_{2}(x,t;k)e^{i \alpha k(x+4k^2t)\widehat{\Lambda}}S(k),
\end{equation}
where $S(k)$ is a $3\times3$-dimensional scattering matrix, and $S(k)=[s_{ij}]_{3\times3}$.

\begin{proposition}[Symmetry]
There exist some symmetrical properties of the eigenfunctions $J_{1,2}(x,t;k)$ and the scattering matrix $S(k)$
\begin{itemize}
\item $J_i^{-1}(x,t;k)=\Sigma\overline{J_i(x,t;\bar{k})}^{\top}\Sigma^{-1}$, $i=1,2$,
\item $S^{-1}(k)=\Sigma\overline{S(\bar{k})}^{\top}\Sigma^{-1}$,
\end{itemize}
where notations $-$ and $\top$ respectively indicate the conjugate and the transpose of the corresponding matrix function, and $\Sigma=\operatorname{diag}(-1,\frac{1}{\alpha},\frac{1}{\alpha})$.
\end{proposition}

\begin{proof}
Combining with the matrix $\Sigma$, the matrix function $Q(x,t)$ satisfies
\begin{equation}
  Q(x,t)=-\Sigma Q^{\dag}(x,t)\Sigma^{-1},\notag
\end{equation}
where the notation $\dag$ represents the Hermitian conjugate of the matrix function $Q$. We replace $k$ with $\bar{k}$ and perform the conjugate transpose of the spatial matrix spectral problem in (\ref{JLax})
\begin{equation}\label{eq1}
\overline{J_x(x,t;\bar{k})}^{\top}+i \alpha k\left[\overline{J(x,t;\bar{k})}^{\top},\Lambda\right]=\overline{J(x,t;\bar{k})}^{\top}Q^{\dag}.
\end{equation}
Then, we multiply $\Sigma$ and $\Sigma^{-1}$ on both sides of the above equation (\ref{eq1}) respectively
\begin{equation}\label{eq2}
\left(\Sigma\overline{J(x,t;\bar{k})}^{\top}\Sigma^{-1}\right)_x-i \alpha k\left[\Lambda, \Sigma\overline{J(x,t;\bar{k})}^{\top}\Sigma^{-1}\right]=-\Sigma\overline{J(x,t;\bar{k})}^{\top}\Sigma^{-1}Q.
\end{equation}
We observe that the equation satisfied by $J^{-1}(x,t;k)$ has the same form as the above equation (\ref{eq2}), thus
\begin{equation}
J^{-1}(x,t;k)=\Sigma\overline{J(x,t;\bar{k})}^{\top}\Sigma^{-1},\notag
\end{equation}
which proves the first symmetrical property. According to the linear relationship (\ref{S}) between the eigenfunctions $J_{1}(x,t;k)$ and $J_{2}(x,t;k)$, the scattering matrix $S(k)$ satisfies the second symmetrical property.
\end{proof}

\begin{proposition}[Reconstruction of solutions]
Suppose that the matrix function $J(x,t;k)$ admits the large-$k$ expansion
\begin{equation}\label{expansion}
  J(x,t;k)=J^{(0)}(x,t)+\frac{J^{(1)}(x,t)}{k}+\frac{J^{(2)}(x,t)}{k^2}+O(k^{-3}),\ k\rightarrow\infty,
\end{equation}
thus the potentials $u(x,t)$ and $v(x,t)$ can be recovered from the matrix function $J(x,t;k)$ by
\begin{gather}\label{sol}
\begin{split}
  u(x,t)=2i \lim\limits_{k\to\infty}k J_{12}(x,t;k), \\
  v(x,t)=2i \lim\limits_{k\to\infty}k J_{13}(x,t;k),
\end{split}
\end{gather}
where the subscript 'ij' stands for the element in row i and column j of the corresponding matrix.
\end{proposition}

\begin{proof}
Substituting the asymptotic expansion (\ref{expansion}) into the spatial part of the Lax pair (\ref{JLax}) of the matrix function $J(x,t;k)$, and comparing powers of $k$, we obtain at orders $O(k)$ and $O(1)$ as
\begin{gather}
  i \alpha \left[\Lambda,J^{(0)}\right]=0,\notag \\
  J^{(0)}_x=i \alpha\left[\Lambda,J^{(1)}\right]+QJ^{(0)}.\notag
\end{gather}
Similarly, substituting into the temporal part, the coefficients of orders $O(k^3)$ and $O(k^2)$ are derived as
\begin{gather}
  4i \alpha \left[\Lambda,J^{(0)}\right]=0,\notag \\
  4i \alpha \left[\Lambda,J^{(1)}\right]+4QJ^{(0)}=0.\notag
\end{gather}
Thus, $J^{(0)}_x=0$. Taking into account the asymptotic conditions of $J(x,t;k)$, we know $J^{(0)}=I$. Then
\begin{gather}
\begin{aligned}
Q&=-i \alpha\left[\Lambda,J^{(1)}\right]\\
&=-2i \alpha\left(\begin{array}{ccc}
0 & -J_{12}^{(1)} & -J_{13}^{(1)} \\
J_{21}^{(1)} & 0 & 0\\
J_{31}^{(1)} & 0 & 0
\end{array}\right).\notag
\end{aligned}
\end{gather}
Considering the expression of matrix $Q$, the potentials $u(x,t)$ and $v(x,t)$ can be constructed as
\begin{equation}
  u(x,t)=2i J_{12}^{(1)}(x,t),\quad v(x,t)=2i J_{13}^{(1)}(x,t),\notag
\end{equation}
which are equivalent to the formulae (\ref{sol}) of the solutions $u(x,t)$ and $v(x,t)$.
\end{proof}

We introduce the definition $\theta(k)=i \alpha k(x+4k^2t)$, and the linear relationship (\ref{S}) can be rewritten as
\begin{gather}\label{SS}
  (J_{11},J_{12},J_{13})=(J_{21},J_{22},J_{23})e^{\theta(k)\widehat{\Lambda}}\left(\begin{array}{ccc}
  s_{11} & s_{12} & s_{13} \\
  s_{21} & s_{22} & s_{23} \\
  s_{31} & s_{32} & s_{33}
  \end{array}\right).
\end{gather}
Through calculating, the above formula (\ref{SS}) is equivalent to
\begin{gather}\label{SSS}
\begin{aligned}
  \left(\frac{J_{11}}{s_{11}},J_{22},J_{23}\right)=&\left(J_{21},\frac{s_{33}J_{12}-s_{32}J_{13}}{\Delta(k)},\frac{-s_{23}J_{12}+s_{22}J_{13}}{\Delta(k)}\right)\\
  &\left(\begin{array}{ccc}
  \frac{1}{s_{11}\Delta(k)} & e^{-2\theta(k)}\frac{s_{13}s_{32}-s_{12}s_{33}}{\Delta(k)} & e^{-2\theta(k)}\frac{s_{12}s_{23}-s_{13}s_{22}}{\Delta(k)}\\
  e^{2\theta(k)}\frac{s_{21}}{s_{11}} & 1 & 0 \\
  e^{2\theta(k)}\frac{s_{31}}{s_{11}} & 0 & 1
\end{array}\right),
\end{aligned}
\end{gather}
where $\Delta(k)=s_{22}s_{33}-s_{23}s_{32}$. We define
\begin{gather}
  \mathbb{D}_+=\{k\in\mathbb{C}|\alpha\operatorname{Im}(k)>0\},\quad \mathbb{D}_-=\{k\in\mathbb{C}|\alpha\operatorname{Im}(k)<0\},\notag
\end{gather}
and consider the analytical properties of Jost solutions $J_1(x,t;k)$ and $J_2(x,t;k)$. Thus, the $J_{11}$, $J_{22}$ and $J_{23}$ are bounded and analytic in the region $\mathbb{D}_+$, and the $J_{12}$, $J_{13}$ and $J_{21}$ are bounded and analytic in the region $\mathbb{D}_-$. Furthermore, the matrix function $M(x,t;k)$ is defined by
\begin{gather}\label{M}
M(x,t;k)=\left\{\begin{array}{l}
\left(\frac{J_{11}}{s_{11}},J_{22},J_{23}\right),\ k\in \mathbb{D}_+,\\
\left(J_{21},\frac{s_{33}J_{12}-s_{32}J_{13}}{\Delta(k)},\frac{-s_{23}J_{12}+s_{22}J_{13}}{\Delta(k)}\right),\  k\in \mathbb{D}_-,
\end{array}\right.
\end{gather}
Combining the relational expression (\ref{SSS}) with the definition (\ref{M}) of $M(x,t;k)$, the Riemann-Hilbert problem of the coupled mKdV equation (\ref{cmKdV}) is constructed as

\noindent\textbf{Riemann-Hilbert problem}
\begin{itemize}
\item $M(x,t;k)$ is meromorphic for $k\in \mathbb{C}\setminus\mathbb{R}$;
\item $M_{+}(x,t;k)=M_{-}(x,t;k)G(x,t;k),\ k\in\mathbb{R},$

where the jump matrix is
\begin{gather}
G(x,t;k)=
\left(\begin{array}{ccc}
  \frac{1}{s_{11}\Delta(k)} & e^{-2\theta(k)}\frac{s_{13}s_{32}-s_{12}s_{33}}{\Delta(k)} & e^{-2\theta(k)}\frac{s_{12}s_{23}-s_{13}s_{22}}{\Delta(k)}\\
  e^{2\theta(k)}\frac{s_{21}}{s_{11}} & 1 & 0 \\
  e^{2\theta(k)}\frac{s_{31}}{s_{11}} & 0 & 1
\end{array}\right);\notag
\end{gather}
\item $M(x,t;k)\rightarrow I,\ k\rightarrow\infty$.
\end{itemize}

\begin{remark}
Since $\Delta(k)$ is the first element in the first column of the inverse of $S(k)$, that is
\begin{equation}
  \Delta(k)=(S^{-1}(k))_{11}=\overline{s_{11}(\bar{k})}.\notag
\end{equation}
Then, if $s_{11}(k)$ exists the simple zeros $\{k_j\}_{j=1}^N\in\mathbb{D}_+$, the $\Delta(k)$ has the simple zeros $\{\bar{k}_j\}_{j=1}^N\in\mathbb{D}_-$. At this time, the Riemann-Hilbert problem satisfies
the pole conditions
\begin{itemize}
\item $\operatorname{Res}\{M(x,t;k),k=k_j\}=\lim\limits_{k\to k_j}M(x,t;k)\left(\begin{array}{ccc}
0 & 0& 0\\
c_je^{2\theta(k_j)} & 0 & 0\\
d_je^{2\theta(k_j)} & 0 & 0
\end{array}\right)$;
\item $\operatorname{Res}\{M(x,t;k),k=\bar{k}_j\}=\lim\limits_{k\to \bar{k}_j}M(x,t;k)\left(\begin{array}{ccc}
    0 & -\bar{c}_je^{-2\theta(\bar{k}_j)} & -\bar{d}_je^{-2\theta(\bar{k}_j)}\\
    0& 0 & 0\\
    0& 0 & 0
    \end{array}\right)$.
\end{itemize}
\end{remark}

\begin{proposition}
The solutions of the coupled mKdV equation (\ref{cmKdV}) are obtained by the above Riemann-Hilbert problem
\begin{gather}\label{solution}
\begin{split}
  u(x,t)=2i \lim\limits_{k\to\infty}kM(x,t;k)_{12}, \\
  v(x,t)=2i \lim\limits_{k\to\infty}kM(x,t;k)_{13},
\end{split}
\end{gather}
where the subscripts 12 and 13 respectively represent the elements in the second column and the third column of the first row of the matrix $M(x,t;k)$.
\end{proposition}

\section{Numerical direct scattering}\label{sec3}

In this section, we numerically investigate the scattering data in the direct scattering problem associated with the coupled mKdV system (\ref{cmKdV}). Based on the spatial part of the Lax pair, we develop the Chebyshev collocation method for computing the scattering matrix $S(k)$ and eigenvalue $k$. These scattering data enter directly into the jump condition for the Riemann-Hilbert formulation, thus their accurate numerical computation is essential for the subsequent implementation of the NIST method.

Since the scattering matrix $S(k)$ satisfies the expression (\ref{S}) and depends only on the spectral parameter $k$ not on $x$ or $t$, $S(k)$ can be determined by setting $x=t=0$
\begin{gather}\label{Sexpression}
S(k)=J_2^{-1}(0,0;k)J_1(0,0;k).
\end{gather}
The matrix functions $J_1(x,t;k)$ and $J_2(x,t;k)$ respectively exhibit asymptotic properties $J_1\rightarrow I,\ x\rightarrow-\infty$ and $J_2\rightarrow I,\ x\rightarrow+\infty$. We define two new matrix functions
\begin{gather}\label{PhiJ}
\Phi_i(x,t;k)=J_i(x,t;k)-I,\ i=1,2,
\end{gather}
which possess the asymptotic properties $\Phi_1\rightarrow0,\ x\rightarrow-\infty$ and $\Phi_2\rightarrow0,\ x\rightarrow+\infty$, respectively. Not only that, the $\Phi_1(x,t;k)$ and $\Phi_2(x,t;k)$ are matrix solutions of the new equation with $x$ partial derivative
\begin{gather}\label{Phi}
\Phi_x-i \alpha k[\Lambda,\Phi]-Q\Phi=Q,
\end{gather}
where
\begin{gather}
\Phi=\left(\begin{array}{ccc}
\phi_{11} & \phi_{12} & \phi_{13} \\
\phi_{21} & \phi_{22} & \phi_{23} \\
\phi_{31} & \phi_{32} & \phi_{33}
\end{array}\right).\notag
\end{gather}
Then, the equation (\ref{Phi}) satisfied by $\Phi_1(x,t;k)$ and $\Phi_2(x,t;k)$ is equivalent to the following three matrix equations
\begin{gather}\label{Phi1}
\left(\begin{array}{c}
\phi_{11x}  \\
\phi_{21x}  \\
\phi_{31x}
\end{array}\right)-2i \alpha k\left(\begin{array}{c}
0  \\
\phi_{21}  \\
\phi_{31}
\end{array}\right)-\left(\begin{array}{c}
\alpha(u\phi_{21}+v\phi_{31}) \\
u\phi_{11} \\
v\phi_{11}
\end{array}\right)=\left(\begin{array}{c}
0 \\
u \\
v
\end{array}\right),
\end{gather}
\begin{gather}\label{Phi2}
\left(\begin{array}{c}
\phi_{12x}  \\
\phi_{22x}  \\
\phi_{32x}
\end{array}\right)+2i \alpha k\left(\begin{array}{c}
\phi_{12}  \\
0 \\
0
\end{array}\right)-\left(\begin{array}{c}
\alpha(u\phi_{22}+v\phi_{32}) \\
u\phi_{12} \\
v\phi_{12}
\end{array}\right)=\left(\begin{array}{c}
\alpha u \\
0 \\
0
\end{array}\right),
\end{gather}
\begin{gather}\label{Phi3}
\left(\begin{array}{c}
\phi_{13x}  \\
\phi_{23x}  \\
\phi_{33x}
\end{array}\right)+2i \alpha k\left(\begin{array}{c}
\phi_{13}  \\
0 \\
0
\end{array}\right)-\left(\begin{array}{c}
\alpha(u\phi_{23}+v\phi_{33}) \\
u\phi_{13} \\
v\phi_{13}
\end{array}\right)=\left(\begin{array}{c}
\alpha v \\
0 \\
0
\end{array}\right).
\end{gather}
By applying the Chebyshev collocation method \citep{Akyuz1999, Cui2024-2}, the systems (\ref{Phi1})-(\ref{Phi3}) can be converted into algebraic equations and then solved numerically. To this end, we begin to introduce the $n$ Chebyshev nodes and the associated Chebyshev polynomials
\begin{gather}
\boldsymbol{x}=\left(-1,\ \operatorname{cos}\left(\frac{n-2}{n-1}\pi\right),\  \cdots,\  \operatorname{cos}\left(\frac{1}{n-1}\pi\right),\ 1\right)^{\top},\notag\\
T_m(x)=\operatorname{cos}(m\operatorname{arccos}x),\ m=0,\ 1,\ \cdots,\ n-1.\notag
\end{gather}
According to the different asymptotic conditions of $\Phi_1$ and $\Phi_2$, we select the appropriate mapping function $H(x)$, which maps the interval $(-\infty,0]$ or $[0,+\infty)$ onto the unit interval $\mathbb{I}$
\begin{gather}
for\  \Phi_1,\ H(x)=2\operatorname{tanh}(ax)+1:\ (-\infty,0]\rightarrow\mathbb{I},\notag\\
for\  \Phi_2,\ H(x)=2\operatorname{tanh}(ax)-1:\ [0,+\infty)\rightarrow\mathbb{I}.\notag
\end{gather}
Focusing on the function $\phi(x)$ defined on the interval $(-\infty,0]$ or $[0,+\infty)$, we represent $\phi(x)$ using Chebyshev polynomials $T_m(x)$, mapping function $H(x)$, and Chebyshev nodes $\boldsymbol{x}$
\begin{gather}\label{phi0}
\phi(x)=T^{\mathbb{R}}(x)\mathcal{F}\phi(\hat{\boldsymbol{x}}),
\end{gather}
where
\begin{gather}
T^{\mathbb{R}}(x)=(T_0(H(x)),\ T_1(H(x)),\ \cdots,\ T_{n-1}(H(x))),\notag\\
\mathcal{F}=(T_0(\boldsymbol{x}),\ T_1(\boldsymbol{x}),\cdots,\ T_{n-1}(\boldsymbol{x}))^{-1},\ \hat{\boldsymbol{x}}=H^{-1}(\boldsymbol{x}).\notag
\end{gather}
Not only that, the derivative of function $\phi(x)$ is expressed based on the derivative properties of Chebyshev polynomials as
\begin{gather}\label{dphi}
\frac{\operatorname{d}\phi(x)}{\operatorname{d}x}=\frac{\operatorname{d}H(x)}{\operatorname{d}x}T^{\mathbb{R}}(x)\mathcal{D}\mathcal{F}\phi(\hat{\boldsymbol{x}}),
\end{gather}
where
\begin{gather}
\mathcal{D}=\left(\begin{array}{cccccc}
0 & 1 & 0 & 3 & \cdots & 0\\
0 & 0 & 4 & 0 & \cdots & 2(n-1)\\
0 & 0 & 0 & 6 & \cdots & 0\\
\vdots & \vdots & \vdots & \vdots & \ddots & \vdots\\
0 & 0 & 0 & 0 & \cdots & 2(n-1)\\
0 & 0 & 0 & 0 & \cdots & 0
\end{array}\right),\quad n\ is\ odd,\notag\\
\mathcal{D}=\left(\begin{array}{cccccc}
0 & 1 & 0 & 3 & \cdots & n-1\\
0 & 0 & 4 & 0 & \cdots & 0\\
0 & 0 & 0 & 6 & \cdots & 2(n-1)\\
\vdots & \vdots & \vdots & \vdots & \ddots & \vdots\\
0 & 0 & 0 & 0 & \cdots & 2(n-1)\\
0 & 0 & 0 & 0 & \cdots & 0
\end{array}\right),\quad  n\ is\ even.\notag
\end{gather}
Combining with the expressions (\ref{phi0}) and (\ref{dphi}), the above matrix equations (\ref{Phi1})-(\ref{Phi3}) can be rewritten as the following discrete form
\begin{gather}\label{phi11}
\left(\begin{array}{ccc}
\mathcal{H} & -\alpha \operatorname{D}[u(\hat{\boldsymbol{x}})] & -\alpha \operatorname{D}[v(\hat{\boldsymbol{x}})] \\
-\operatorname{D}[u(\hat{\boldsymbol{x}})] & \mathcal{H}-2i \alpha kI & 0 \\
-\operatorname{D}[v(\hat{\boldsymbol{x}})] & 0 & \mathcal{H}-2i \alpha kI
\end{array}\right)\left(\begin{array}{c}
\phi_{11}(\hat{\boldsymbol{x}})  \\
\phi_{21}(\hat{\boldsymbol{x}})  \\
\phi_{31}(\hat{\boldsymbol{x}})
\end{array}\right)=\left(\begin{array}{c}
0 \\
u(\hat{\boldsymbol{x}}) \\
v(\hat{\boldsymbol{x}})
\end{array}\right),
\end{gather}
\begin{gather}\label{phi21}
\left(\begin{array}{ccc}
\mathcal{H}+2i \alpha kI & -\alpha \operatorname{D}[u(\hat{\boldsymbol{x}})] & -\alpha \operatorname{D}[v(\hat{\boldsymbol{x}})] \\
-\operatorname{D}[u(\hat{\boldsymbol{x}})] & \mathcal{H} & 0 \\
-\operatorname{D}[v(\hat{\boldsymbol{x}})] & 0 & \mathcal{H}
\end{array}\right)\left(\begin{array}{c}
\phi_{12}(\hat{\boldsymbol{x}})  \\
\phi_{22}(\hat{\boldsymbol{x}})  \\
\phi_{32}(\hat{\boldsymbol{x}})
\end{array}\right)=\left(\begin{array}{c}
\alpha u(\hat{\boldsymbol{x}})\\
0 \\
0
\end{array}\right),
\end{gather}
\begin{gather}\label{phi31}
\left(\begin{array}{ccc}
\mathcal{H}+2i \alpha kI & -\alpha \operatorname{D}[u(\hat{\boldsymbol{x}})] & -\alpha \operatorname{D}[v(\hat{\boldsymbol{x}})] \\
-\operatorname{D}[u(\hat{\boldsymbol{x}})] & \mathcal{H} & 0 \\
-\operatorname{D}[v(\hat{\boldsymbol{x}})] & 0 & \mathcal{H}
\end{array}\right)\left(\begin{array}{c}
\phi_{13}(\hat{\boldsymbol{x}})  \\
\phi_{23}(\hat{\boldsymbol{x}})  \\
\phi_{33}(\hat{\boldsymbol{x}})
\end{array}\right)=\left(\begin{array}{c}
\alpha v(\hat{\boldsymbol{x}})\\
0 \\
0
\end{array}\right),
\end{gather}
where $\mathcal{H}=\operatorname{D}\left[\frac{\operatorname{d}H(\hat{\boldsymbol{x}})}{\operatorname{d}x}\right]\mathcal{F}^{-1}\mathcal{D}\mathcal{F}$. The matrix functions $\Phi_1$ and $\Phi_2$ can be determined at the discrete point $k$ by numerically solving the matrix equations (\ref{phi11})-(\ref{phi31}). Subsequently, the value of the scattering matrix $S(k)$ is derived based on the expressions (\ref{Sexpression}) and (\ref{PhiJ}), and the elements $s_{ij}(k)$ can also be obtained. We introduce the row vector $\rho(k)$ of the reflection coefficients
\begin{gather}\label{rho}
\rho(k)=(s_{12}(k),s_{13}(k))\left(\begin{array}{cc}
s_{22}(k) & s_{23}(k) \\
s_{32}(k) & s_{33}(k)
\end{array}\right)^{-1}=(\rho_1(k),\rho_2(k)).
\end{gather}

Figs. \ref{fi1}-\ref{fi2} present the numerically computed reflection coefficients $\rho_1(k)$ and $\rho_2(k)$ for different initial potential functions. In these figures, the calculated results are shown under $a=0.1$ and $n=101$, and the calculated interval is truncated into $\left[-5,5\right]$. The blue line represents the real part of associated reflection coefficient, and the dotted red line denotes the imaginary part. In Fig. \ref{fi1}, reflection coefficients $\rho_1(k)$ and $\rho_2(k)$ are calculated with two different initial profiles $u_0=0.5\operatorname{sech}(x)$ and $v_0=1.5\operatorname{exp}(-x^2)$, which exhibit clear differences in the overall shape and peak distribution. In Fig. \ref{fi2}, we consider a second set of initial potentials $u_0=\operatorname{exp}(-x^2)\operatorname{sech}(x)$ and $v_0=0.5\operatorname{exp}(-x^2)\operatorname{sech}(x)$, and the calculated $\rho_1(k)$ and $\rho_2(k)$ possess similar curve profiles and differ only by the amplitudes. To demonstrate that the numerical scheme applies to both $\alpha=-1$ and $\alpha=1$, we also compute reflection coefficients for these two signs with the same initial profiles in Figs. \ref{fi1}-\ref{fi2}. It can be seen that both $\rho_1(k)$ and $\rho_2(k)$ smoothly distribute near the origin of $k$ and decay rapidly as $|k|$ increases. Moreover, changing the sign of $\alpha$ leads to a noticeable change in the profiles of both the real and imaginary parts.

\begin{figure}[ht]
\centering
  \includegraphics[width=0.75\linewidth]{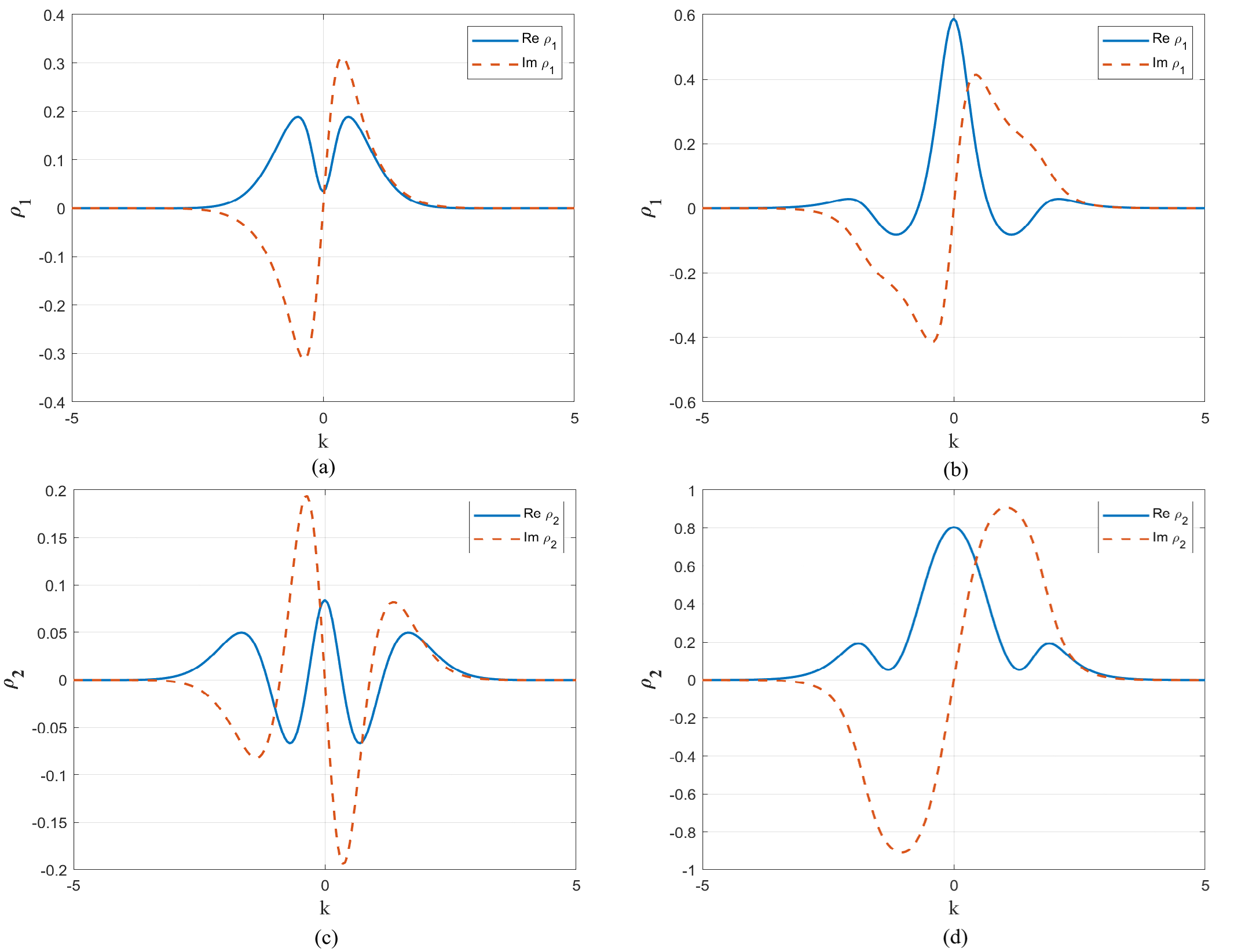}
  \caption{The calculated reflection coefficients $\rho_1(k)$ and $\rho_2(k)$ for $u_0=0.5\operatorname{sech}(x)$ and $v_0=1.5\operatorname{exp}(-x^2)$ initial potentials. (a)-(c): $\alpha=-1$; (b)-(d): $\alpha=1$.}
\label{fi1}
\end{figure}

\begin{figure}[ht]
\centering
  \includegraphics[width=0.75\linewidth]{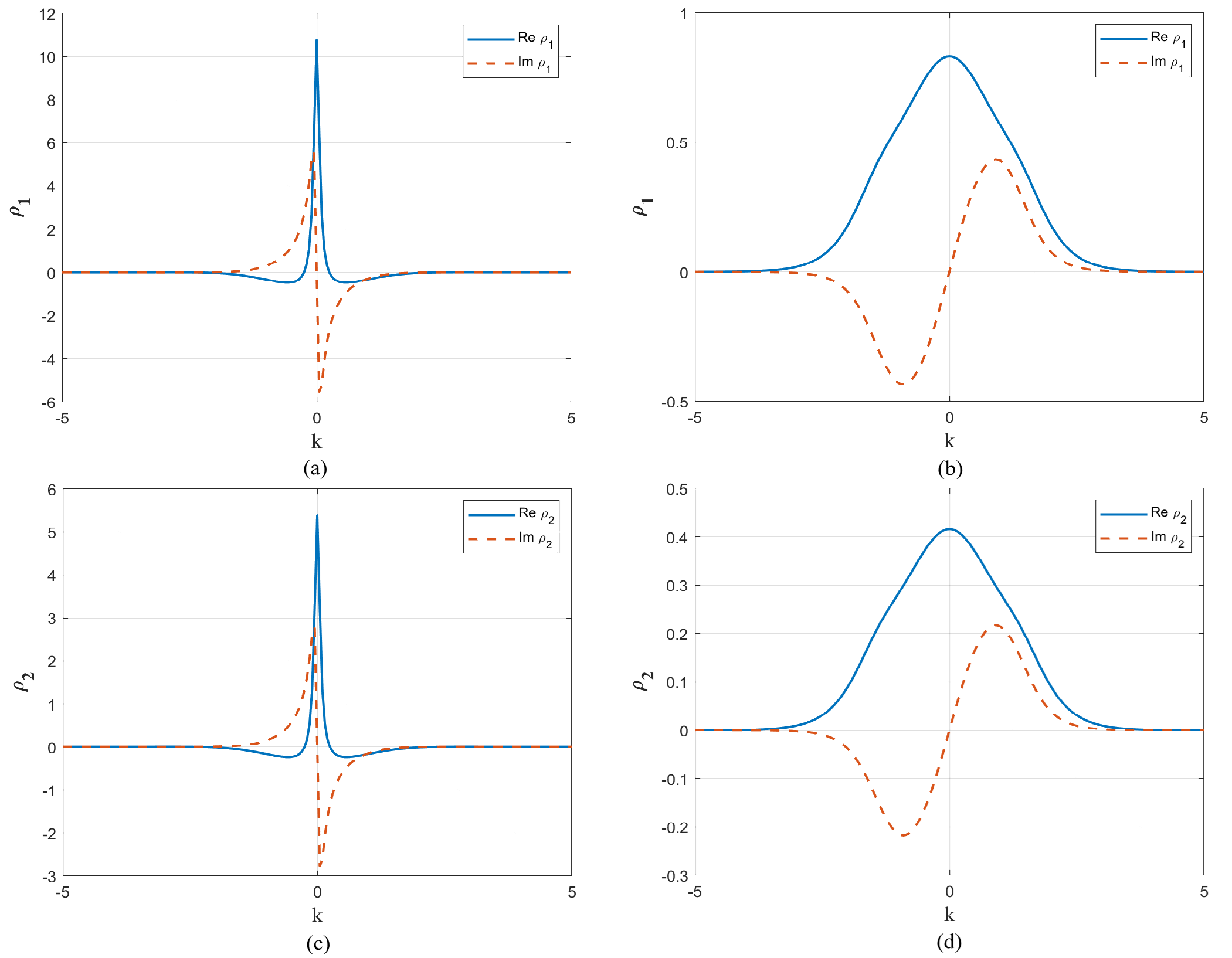}
  \caption{The calculated reflection coefficients $\rho_1(k)$ and $\rho_2(k)$ for $u_0=\operatorname{exp}(-x^2)\operatorname{sech}(x)$ and $v_0=0.5\operatorname{exp}(-x^2)\operatorname{sech}(x)$ initial potentials. (a)-(c): $\alpha=-1$; (b)-(d): $\alpha=1$.}
\label{fi2}
\end{figure}

The coupled mKdV equation (\ref{cmKdV}) with parameter $\alpha=-1$ has the following form of the line soliton solutions \citep{Wu2017}
\begin{gather}\label{linesol}
\begin{split}
  u(x,t)=-\frac{2\sqrt{3}}{3}\operatorname{sech}(2x-8t-\operatorname{ln}\frac{\sqrt{3}}{2}), \\
  v(x,t)=-\frac{2\sqrt{6}}{3}\operatorname{sech}(2x-8t-\operatorname{ln}\frac{\sqrt{3}}{2}),
\end{split}
\end{gather}
thus, the initial potential functions of the above line soliton solutions (\ref{linesol}) are
\begin{gather}\label{inilinesol}
\begin{split}
  u_0(x)=-\frac{2\sqrt{3}}{3}\operatorname{sech}(2x-\operatorname{ln}\frac{\sqrt{3}}{2}), \\
  v_0(x)=-\frac{2\sqrt{6}}{3}\operatorname{sech}(2x-\operatorname{ln}\frac{\sqrt{3}}{2}).
\end{split}
\end{gather}
In this case, the reflection coefficients $\rho_1(k)$ and $\rho_2(k)$ remain close to zero in $\mathbb{R}$. To investigate the influence of mapping parameter $a$ on numerical results, we fix the number of collocation points at $n=101$ and compute over the spectral interval $k\in[-10,10]$. The absolute values of $\rho_1(k)$ and $\rho_2(k)$ are numerically computed with different $a$ from 0.01 to 0.07, which are shown in Fig. \ref{fi3}. It is obvious that if the mapping parameter $a\in(0,0.1]$, the calculated results become better as $a$ increases.

\begin{figure}[ht]
\centering
  \includegraphics[width=0.75\linewidth]{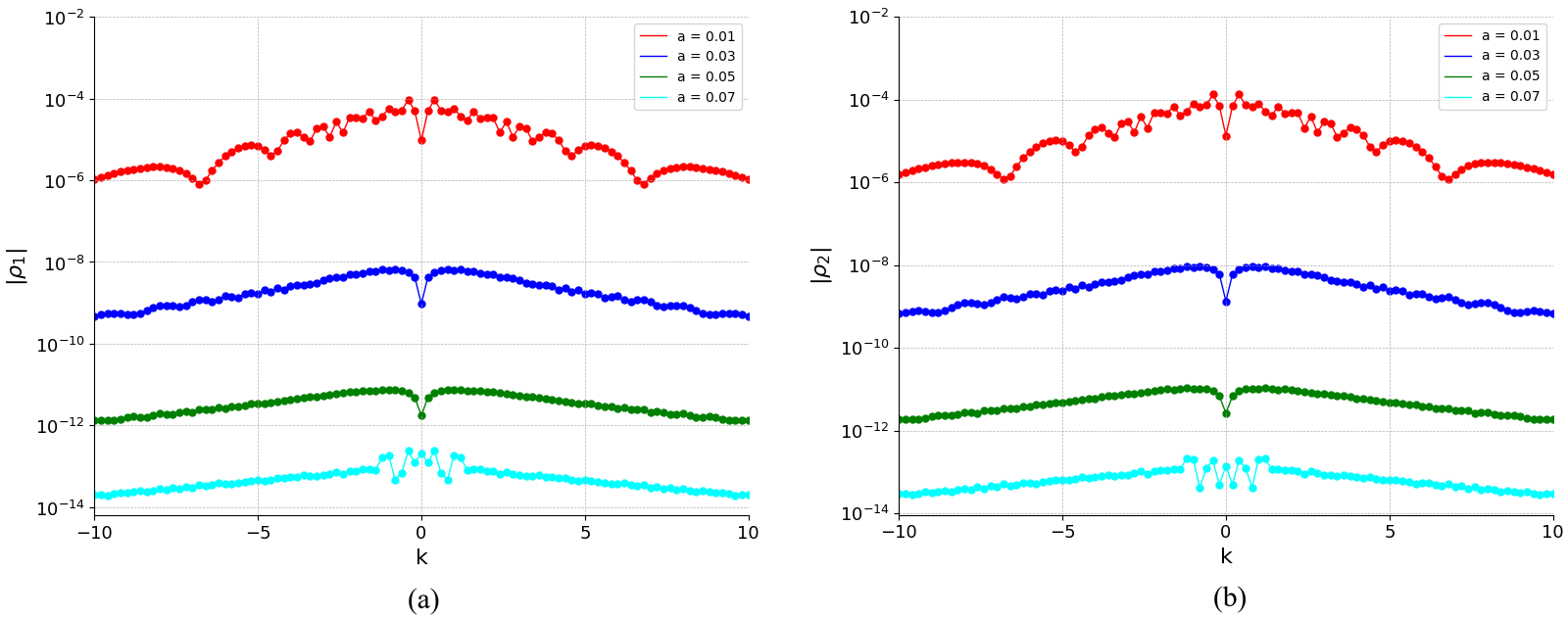}
  \caption{The absolute values of $\rho_1(k)$ and $\rho_2(k)$ for initial potentials (\ref{inilinesol}) of line soliton solutions (\ref{linesol}).}
\label{fi3}
\end{figure}

To assess the numerical stability and convergence of the numerical scheme with respect to the collocation number $n$, we fix $a=0.1$, $\alpha=1$ and $k=0$, and choose the initial conditions as $u_0=\operatorname{exp}(-x^2)\operatorname{sech}(x)$ and $v_0=0.5\operatorname{exp}(-x^2)\operatorname{sech}(x)$. The reflection coefficients obtained with $n_{\operatorname{ref}}=301$ are taken as the references, and we define
\begin{gather}
  \operatorname{Error}1(n)=|\rho_1(n,0)-\rho_1(301,0)|,\quad \operatorname{Error}2(n)=|\rho_2(n,0)-\rho_2(301,0)|.\notag
\end{gather}
Fig. \ref{fi4} displays that both errors decay rapidly as $n$ increases from small values, and then level off around $10^{-12}$. The experiment indicates that the proposed Chebyshev collocation method is numerically stable and convergent for solving the reflection coefficient $\rho(k)$.

\begin{figure}[ht]
\centering
  \includegraphics[width=0.75\linewidth]{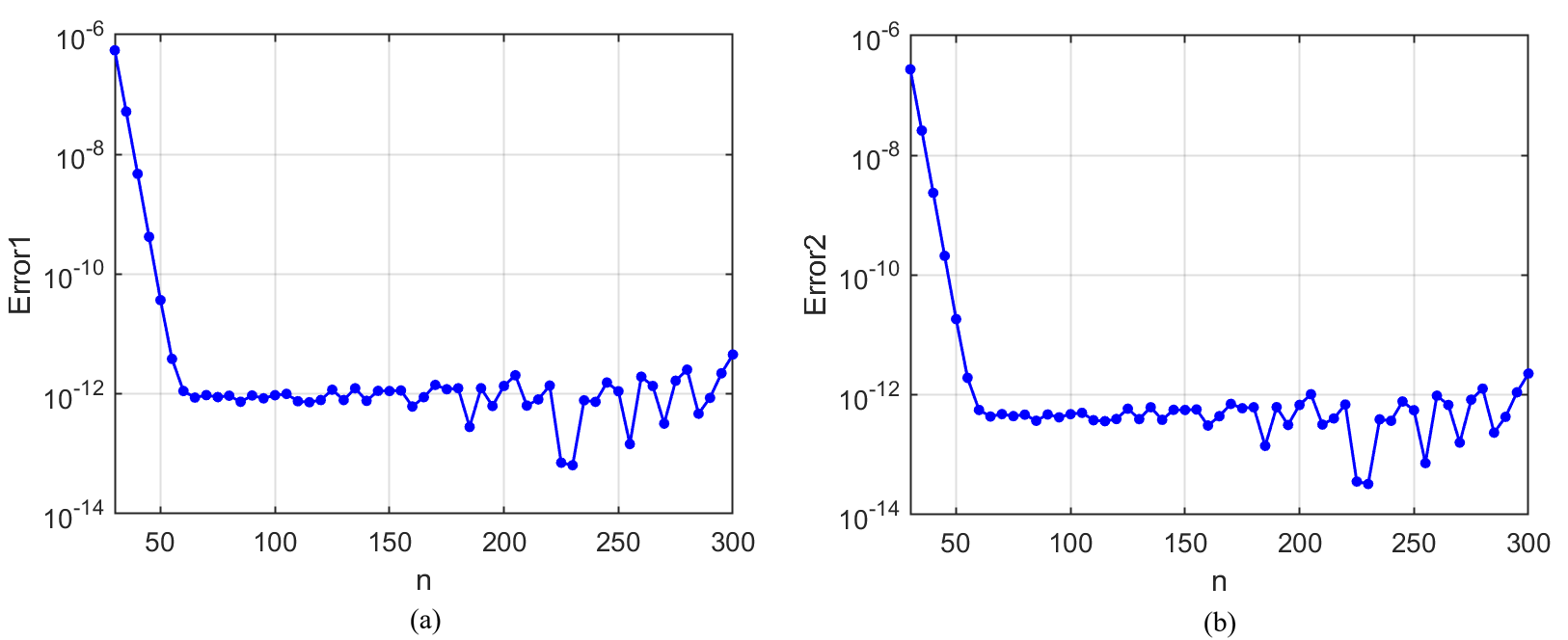}
  \caption{The convergence of reflection coefficients $\rho_1(k)$ and $\rho_2(k)$ for $u_0=\operatorname{exp}(-x^2)\operatorname{sech}(x)$ and $v_0=0.5\operatorname{exp}(-x^2)\operatorname{sech}(x)$ initial potentials.}
\label{fi4}
\end{figure}

The eigenvalue $k$ of the coupled mKdV equation (\ref{cmKdV}) can be numerically calculated based on the spatial matrix spectral problem of the Lax pair (\ref{Lax}). We take the eigenfunction
\begin{gather}
  Y(x,t;k)=(Y_1(x,t;k),\ Y_2(x,t;k),\ Y_3(x,t;k))^{\top},\notag
\end{gather}
then the matrix equation with $x$ derivative for $Y(x,t;k)$ is given as
\begin{gather}\label{Y}
  \left(\begin{array}{ccc}
  -\partial x & \alpha u &\alpha v\\
  -u & \partial x & 0 \\
  -v & 0 & \partial x
  \end{array}\right)\left(\begin{array}{c}
		Y_{1}  \\
		Y_{2}  \\
        Y_{3}
	\end{array}\right)=i \alpha k \left(\begin{array}{c}
		Y_{1}  \\
		Y_{2}  \\
        Y_{3}
	\end{array}\right).
\end{gather}
At this moment, we choose the appropriate mapping function $H(x)$ as
\begin{gather}
  H(x)=\operatorname{tanh}(0.1x),\notag
\end{gather}
which maps the interval $\mathbb{R}$ to the unit interval $\mathbb{I}$. Then, the eigenvalue problem (\ref{Y}) of the coupled mKdV equation (\ref{cmKdV}) can be expressed as the following matrix equation
\begin{gather}\label{eigenvalue}
\left(\begin{array}{ccc}
  -\mathcal{H} & \alpha \operatorname{D}[u(\hat{\boldsymbol{x}})] & \alpha \operatorname{D}[v(\hat{\boldsymbol{x}})] \\
  -\operatorname{D}[u(\hat{\boldsymbol{x}})] & \mathcal{H} & 0 \\
  -\operatorname{D}[v(\hat{\boldsymbol{x}})] & 0 & \mathcal{H}
  \end{array}\right)\left(\begin{array}{c}
		Y_{1}(\hat{\boldsymbol{x}})  \\
		Y_{2}(\hat{\boldsymbol{x}})  \\
        Y_{3}(\hat{\boldsymbol{x}})
	\end{array}\right)=i \alpha k\left(\begin{array}{c}
		Y_{1}(\hat{\boldsymbol{x}})  \\
		Y_{2}(\hat{\boldsymbol{x}})  \\
        Y_{3}(\hat{\boldsymbol{x}})
	\end{array}\right).
\end{gather}
After Chebyshev collocation, the spectral problem (\ref{eigenvalue}) is reduced to a finite-dimensional matrix eigenvalue problem, which can be explored numerically using Matlab's eig routine.

Fig. \ref{fi5} illustrates the numerically computed spectral points in the complex $k$-plane under $n=200$ for the initial potentials $u_0=0.5\operatorname{sech}(x)$ and $v_0=1.5\operatorname{exp}(-x^2)$. For $\alpha=-1$, a pair of discrete eigenvalues is located at $k=\pm0.814i $, whereas for $\alpha=1$, no discrete eigenvalue appears. This indicates that the existence of the discrete spectrum depends sensitively on the parameter $\alpha$. Fig. \ref{fi6} further presents the numerical eigenvalues corresponding to the line-soliton initial data (\ref{inilinesol}), where a pair of discrete eigenvalues is found at $k=\pm i $, in agreement with the theoretical prediction. These figures show that the discrete eigenvalues are symmetrically distributed on the upper and lower imaginary axis, while the continuous eigenvalues are located on the real axis.

\begin{figure}[ht]
\centering
  \includegraphics[width=0.75\linewidth]{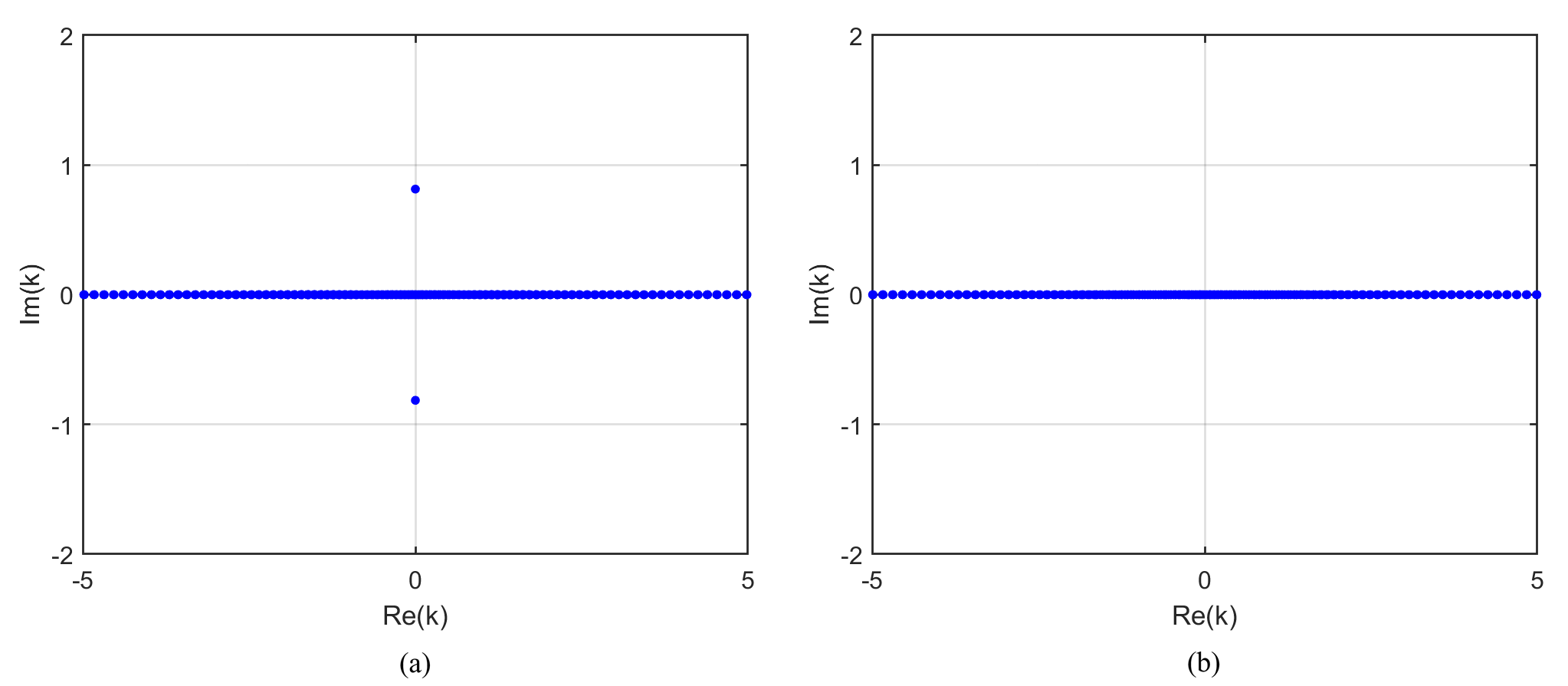}
  \caption{The calculated eigenvalues of initial potentials $u_0=0.5\operatorname{sech}(x)$ and $v_0=1.5\operatorname{exp}(-x^2)$. (a): $\alpha=-1$; (b): $\alpha=1$.}
\label{fi5}
\end{figure}

\begin{figure}[ht]
\centering
  \includegraphics[width=0.45\linewidth]{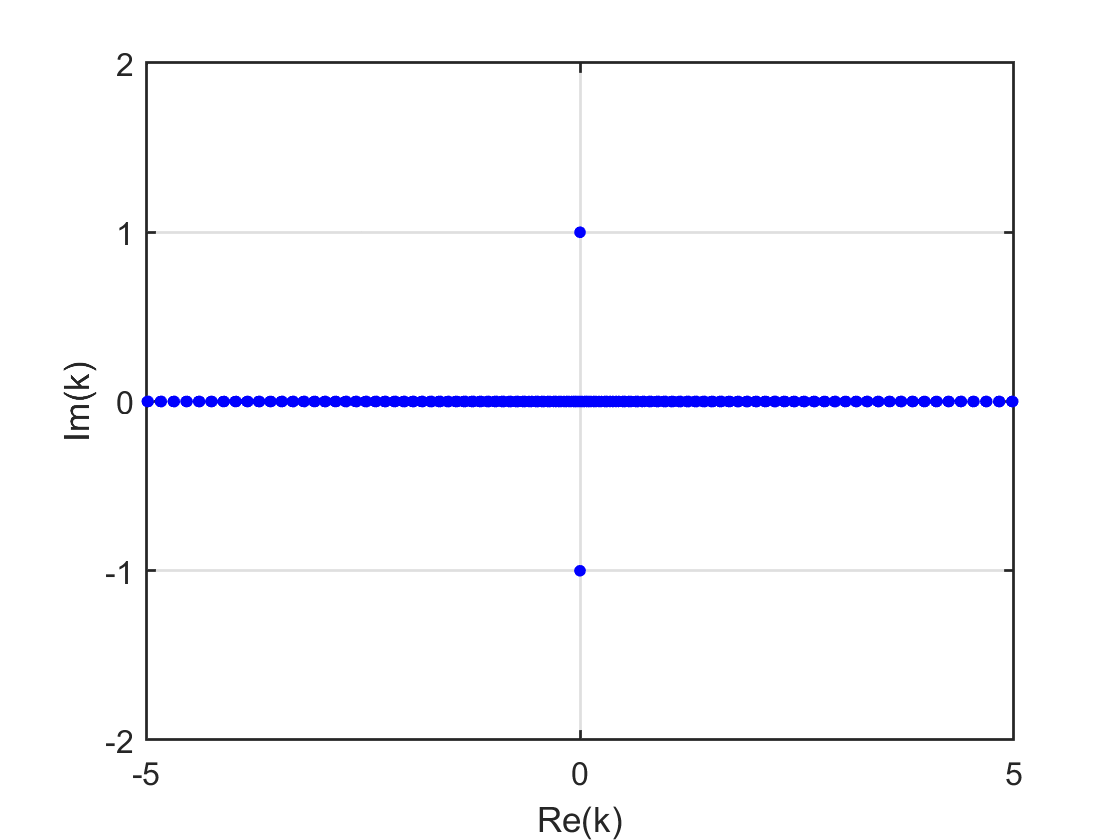}
  \caption{The calculated eigenvalues of the line soliton solutions (\ref{linesol}).}
\label{fi6}
\end{figure}

To test the numerical stability and convergence of the discrete eigenvalue computation, we take the line-soliton initial data (\ref{inilinesol}), and compute the spectral points for a series of $n$ in Fig. \ref{fi7}. We take the result at $n_{\operatorname{ref}}=301$ as the reference value, and define
\begin{gather}
  \operatorname{Error}(n)=|k(n)-k(301)|.\notag
\end{gather}
Fig. \ref{fi7} displays a rapid decay and eventual saturation of this error, which indicates excellent convergence of the discrete eigenvalue computation.

\begin{figure}[ht]
\centering
  \includegraphics[width=0.45\linewidth]{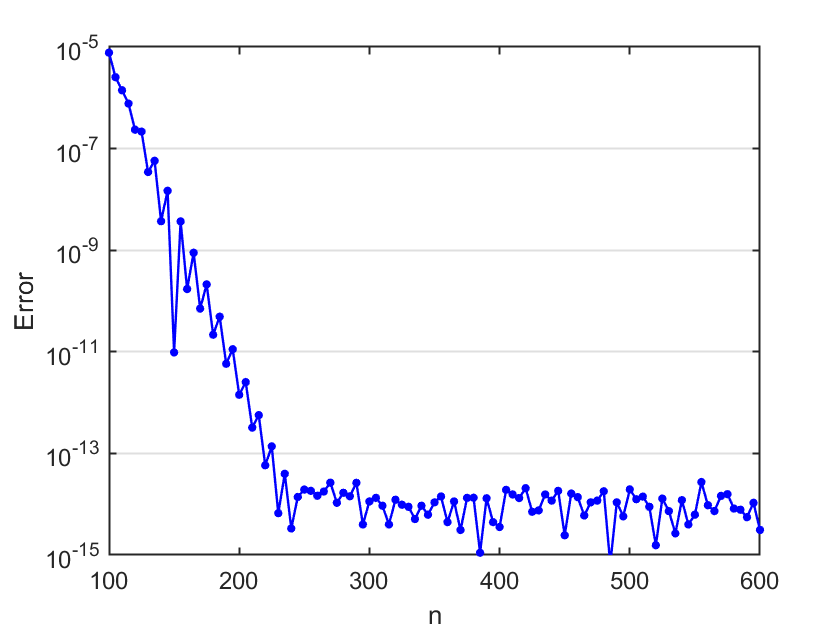}
  \caption{Spectral convergence of the discrete eigenvalues for the line soliton solutions (\ref{linesol}).}
\label{fi7}
\end{figure}

\section{Analysis and deformation of Riemann-Hilbert problem}\label{sec4}

In this section, we apply the Deift-Zhou nonlinear steepest descent method \citep{Deift1993} to the oscillatory Riemann-Hilbert problem formulated above. By factorizing the jump matrix appropriately and deforming the contour, the non-decaying oscillatory factors are redistributed onto contours where they become exponentially decaying. Focusing on the case $\alpha=1$, the jump matrix contains the oscillatory exponential factors $e^{2i k(x+4k^2t)}$ and $e^{-2i k(x+4k^2t)}$. To this end, we introduce the phase function
\begin{gather}\label{psi}
  \psi(k)=2i kx+8i k^3 t.
\end{gather}
Hence the stationary points are determined by $\psi^{\prime}(k)=0$, namely
\begin{gather}
  k_0=\sqrt{-\frac{x}{12t}},\quad -k_0=-\sqrt{-\frac{x}{12t}}.\notag
\end{gather}
The directions of steepest descent through $k_0$ are $\frac{\pi}{4}$ and $\frac{5\pi}{4}$, while the directions of steepest ascent are $-\frac{\pi}{4}$ and $\frac{3\pi}{4}$. In the meantime, the directions of steepest descent through $-k_0$ are $-\frac{\pi}{4}$ and $\frac{3\pi}{4}$, while the directions of steepest ascent are $\frac{\pi}{4}$ and $\frac{5\pi}{4}$.

We take into account $k=\operatorname{Re}k+i \operatorname{Im}k$, then the expression (\ref{psi}) of phase function $\psi(k)$ can be rewritten as
\begin{gather}
  \psi(k)=2i (\operatorname{Re}k+i \operatorname{Im}k)(x+4(\operatorname{Re}^2k+2i \operatorname{Re}k\operatorname{Im}k-\operatorname{Im}^2k)t).\notag
\end{gather}
Thus, combining with stationary point $k_0$, the real part of $\psi(k)$ is given as
\begin{gather}
\operatorname{Re}\psi(k)=8t\operatorname{Im}k(-3\operatorname{Re}^2k+\operatorname{Im}^2k+3k_0^2).\notag
\end{gather}
Since $t>0$, the sign of $\operatorname{Re}\psi(k)$ is determined by the signs of $\operatorname{Im}k$ and $(-3\operatorname{Re}^2k+\operatorname{Im}^2k+3k_0^2)$. Corresponding to the case $x<0$, the phase function $\psi(k)$ has two real stationary points. In this case, the sign distribution of $\operatorname{Re}\psi(k)$ in the complex-$k$ plane is shown in the Fig. \ref{fi8}. For the case $x>0$, the phase function $\psi(k)$ has two imaginary stationary points. Then, Fig. \ref{fi9} presents the sign distribution of $\operatorname{Re}\psi(k)$ of this case. In these two figures, $\operatorname{Re}\psi(k)<0$ in the orange area, and the oscillatory exponential factor $e^{2i k(x+4k^2t)}$ is convergent. In the blue area, $\operatorname{Re}\psi(k)>0$, and the oscillatory exponential factor $e^{-2i k(x+4k^2t)}$ is convergent.

\begin{figure}[ht]
\centering
  \includegraphics[width=0.45\linewidth]{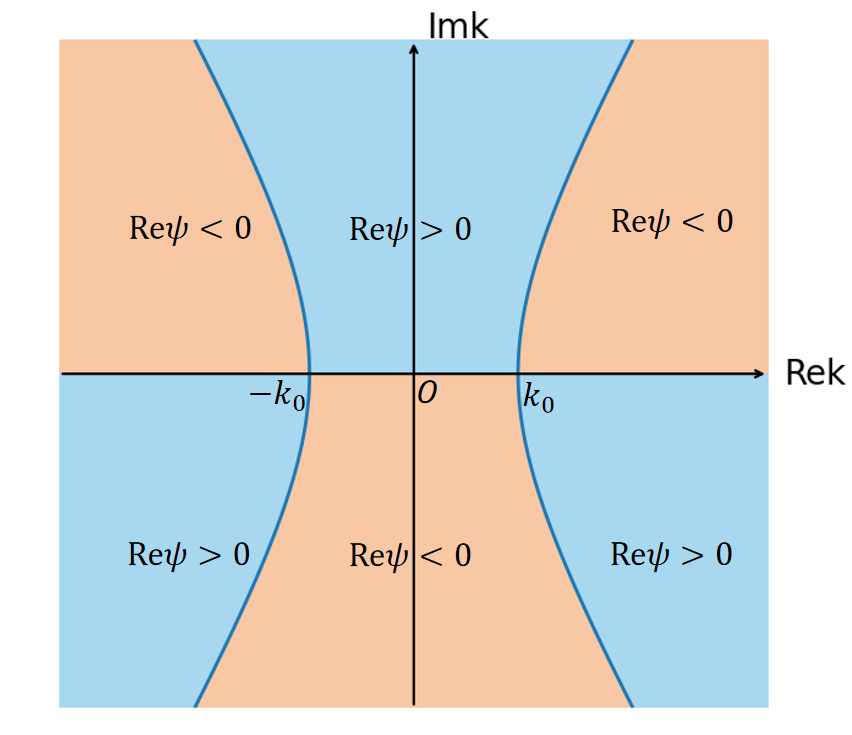}
  \caption{The sign of $\operatorname{Re}\psi(k)$ for $x<0$ in the complex-$k$ plane.}
\label{fi8}
\end{figure}

\begin{figure}[ht]
\centering
  \includegraphics[width=0.45\linewidth]{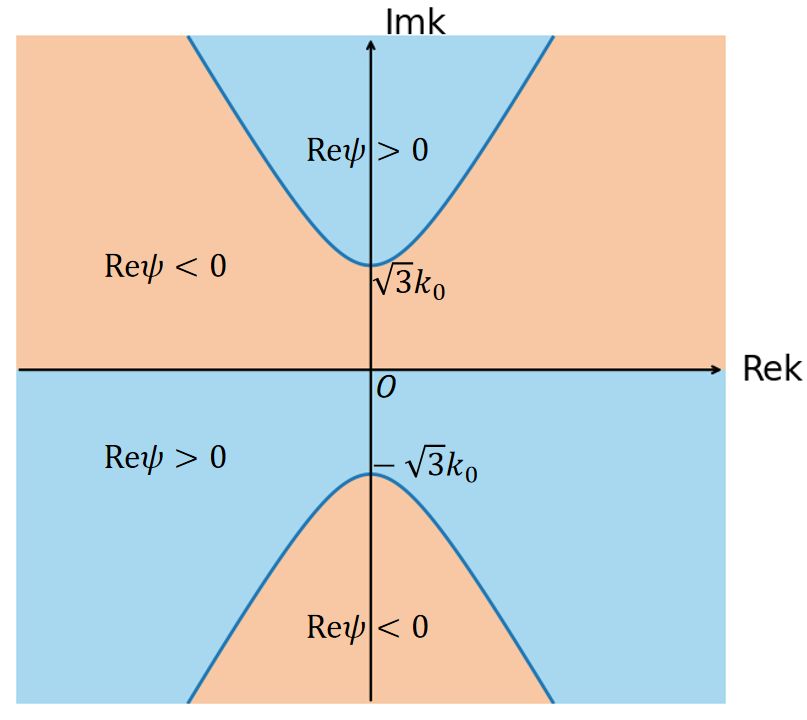}
  \caption{The sign of $\operatorname{Re}\psi(k)$ for $x>0$ in the complex-$k$ plane.}
\label{fi9}
\end{figure}

In order to further deform the original Riemann-Hilbert problem of the coupled mKdV equation (\ref{cmKdV}), we introduce a positive constant $c_1$ and divide the $(x,t)$-plane into the following three regions:
\begin{itemize}
\item The dispersive region
\begin{gather}
-x>c_1t^{1/3},\quad c_1>0.\notag
\end{gather}
In this region, the leading asymptotics of the solution exhibits an oscillatory tail. In fact, the oscillatory behavior is associated with the regime $x/t<0$. In order to separate the transition zone near the origin, we adopt the above refined partition. Since $x<0$, the phase function $\psi(k)$ has two real stationary points in this region.
\item The Painlev\'{e} transition region
\begin{gather}
|x|\leq c_1t^{1/3},\quad c_1>0.\notag
\end{gather}
This region describes the transition from a rapidly decaying regime to an oscillatory regime. It corresponding to the critical scaling near $x=0$, where the stationary points approach and merge at the origin. The leading asymptotics in this region is typically described by a Painlev\'{e}-type special function.
\item The fast decay region
\begin{gather}
x>c_1t^{1/3},\quad c_1>0.\notag
\end{gather}
In this region, since $x>0$, the phase function $\psi(k)$ has no real stationary point, and the solutions $u(x,t)$ and $v(x,t)$ decay rapidly. More precisely, for any $N>0$,
\begin{gather}
u(x,t)=O(t^{-N}),\quad v(x,t)=O(t^{-N}),\notag
\end{gather}
uniformly in this region.
\end{itemize}
In the following, we analyze the corresponding Riemann-Hilbert problem and deform the contour in each of the above regions separately.

\subsection{The dispersive region}

First, we deform the original Riemann-Hilbert problem in the dispersive region. In this region, $\psi(k)$ has two real stationary points, meaning that $k_0$ is a real number. The jump matrix $G(x,t;k)$ of the original Riemann-Hilbert problem contains oscillation terms $\operatorname{exp}(\psi(k))$ and $\operatorname{exp}(-\psi(k))$. In order to separate these two oscillatory exponential factors into different matrixes, we introduce two matrix factorizations of the jump matrix $G(x,t;k)$
\begin{gather}\label{matrixde}
\left\{\begin{array}{l}
G(x,t;k)=L(x,t;k)U(x,t;k),\\
G(x,t;k)=A(x,t;k)B(x,t;k)C(x,t;k),
\end{array}\right.
\end{gather}
where
\begin{gather}
L(x,t;k)=\left(\begin{array}{ccc}
1 & \frac{s_{13}s_{32}-s_{12}s_{33}}{\Delta(k)}e^{-\psi(k)} & \frac{s_{12}s_{23}-s_{13}s_{22}}{\Delta(k)}e^{-\psi(k)}\\
0 & 1 & 0\\
0 & 0 & 1
\end{array}\right),\notag
\end{gather}
\begin{gather}
  U(x,t;k)=\left(\begin{array}{ccc}
  1 & 0 & 0\\
  \frac{s_{21}}{s_{11}}e^{\psi(k)} & 1 & 0\\
  \frac{s_{31}}{s_{11}}e^{\psi(k)} & 0 & 1
  \end{array}\right),\notag
\end{gather}
\begin{gather}
  A(x,t;k)=\left(\begin{array}{ccc}
  1 & 0 & 0\\
  s_{21}\Delta(k)e^{\psi(k)} & 1 & 0\\
  s_{31}\Delta(k)e^{\psi(k)} & 0 & 1
  \end{array}\right),\notag
\end{gather}
\begin{gather}
  B(x,t;k)=\left(\begin{array}{ccc}
  \frac{1}{s_{11}\Delta(k)} & 0 & 0\\
  0 & 1-s_{21}(s_{13}s_{32}-s_{12}s_{33}) & -s_{21}(s_{12}s_{23}-s_{13}s_{22})\\
  0 & -s_{31}(s_{13}s_{32}-s_{12}s_{33}) & 1-s_{31}(s_{12}s_{23}-s_{13}s_{22})
  \end{array}\right),\notag
\end{gather}
\begin{gather}
  C(x,t;k)=\left(\begin{array}{ccc}
  1 & s_{11}(s_{13}s_{32}-s_{12}s_{33})e^{-\psi(k)} & s_{11}(s_{12}s_{23}-s_{13}s_{22})e^{-\psi(k)}\\
  0 & 1 & 0\\
  0 & 0 & 1
  \end{array}\right).\notag
\end{gather}

Through decomposing the matrix $G$ (\ref{matrixde}), the upper triangular factors $L$ and $C$ are extended to region of $\operatorname{Re}\psi(k)>0$, whereas the lower triangular factors $U$ and $A$ are extended to region of $\operatorname{Re}\psi(k)<0$. Taking into account the stationary points $k_0$ and $-k_0$, we divide the complex-$k$ plane into eight regions by means of rapid descent/rise lines and real axis, and the new jump contour is constructed as shown in Fig. \ref{fi10} below.

\begin{figure}[ht]
\centering
  \includegraphics[width=0.9\linewidth]{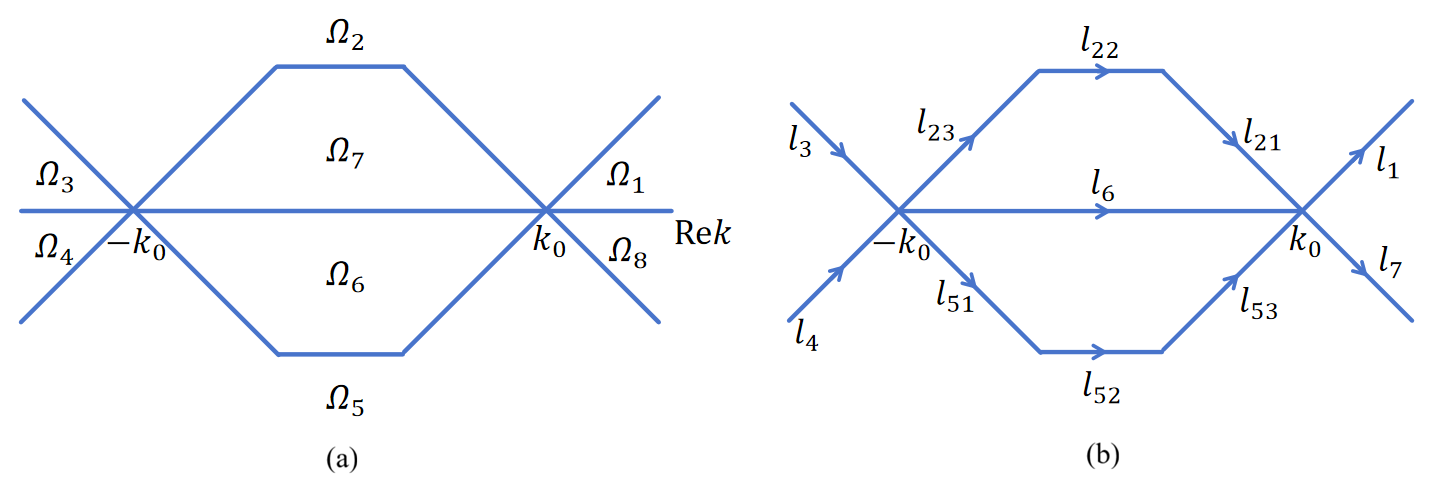}
  \caption{The division of regions and jump contours of Riemann-Hilbert problem I.}
\label{fi10}
\end{figure}

We define a new matrix function $M^{(1)}(x,t;k)$ on these regions
\begin{gather}\label{M1}
M^{(1)}(x,t;k)=M(x,t;k)\left\{\begin{array}{l}
U^{-1},\quad k\in \Omega_1\cup\Omega_3,\\
I,\ \ \ \quad k\in \Omega_2\cup\Omega_5,\\
L,\ \ \ \quad k\in \Omega_4\cup\Omega_8,\\
A,\ \ \ \quad k\in \Omega_6,\\
C^{-1},\quad k\in \Omega_7.
\end{array}\right.
\end{gather}
By construction, the jump matrixes on the newly opened contours contain only one oscillatory exponential factor, which is exponentially close to the identity in the corresponding regions. Then, we derive the following Riemann-Hilbert problem I of $M^{(1)}(x,t;k)$, which is constructed by the relationship (\ref{M1}).

\noindent\textbf{Riemann-Hilbert problem I}
\begin{itemize}
\item $M^{(1)}(x,t;k)$ is analytic in $\mathbb{C}\setminus(l_1\cup ...\cup l_7)$;
\item $M^{(1)}_{+}(x,t;k)=M^{(1)}_{-}(x,t;k)G^{(1)}(x,t;k),\ k\in l_1\cup ...\cup l_7,$

where the new jump matrix $G^{(1)}(x,t;k)$ is
\begin{gather}
G^{(1)}(x,t;k)=\left\{\begin{array}{l}
U,\quad k\in l_1\cup l_3,\\
C,\quad k\in l_{21}\cup l_{22}\cup l_{23},\\
L,\quad k\in l_4\cup l_7,\\
A,\quad k\in l_{51}\cup l_{52}\cup l_{53},\\
B,\quad k\in l_6;
\end{array}\right.\notag
\end{gather}
\item $M^{(1)}(x,t;k)\rightarrow I,\ k\rightarrow\infty$.
\end{itemize}

After the first lens opening, the jump matrix on the interval $(-k_0,k_0)\in\mathbb{R}$ admits a central block matrix $B$, which possesses a $2\times2$ block in the lower-right corner. Then we introduce the $2\times2$ matrix function $\delta(k)$ by the following Riemann-Hilbert problem:
\begin{itemize}
\item $\delta(k)$ is analytic in $\mathbb{C}\setminus\mathbb{R}$;
\item $\delta_{+}(k)=\delta_{-}(k)\left\{\begin{array}{l}
\left(\begin{array}{cc}
  1-s_{21}(s_{13}s_{32}-s_{12}s_{33}) & -s_{21}(s_{12}s_{23}-s_{13}s_{22})\\
  -s_{31}(s_{13}s_{32}-s_{12}s_{33}) & 1-s_{31}(s_{12}s_{23}-s_{13}s_{22})
  \end{array}\right),\ |k|<k_0,\\
  I,\quad |k|>k_0;
\end{array}\right.$
\item $\delta(k)\rightarrow I,\ k\rightarrow\infty$.
\end{itemize}
Meanwhile, the determinant of matrix $\delta(k)$ satisfies the scalar Riemann-Hilbert problem:
\begin{itemize}
\item $\operatorname{det}(\delta(k))$ is analytic in $\mathbb{C}\setminus\mathbb{R}$;
\item $\operatorname{det}(\delta(k))_{+}=\operatorname{det}(\delta(k))_{-}\left\{\begin{array}{l}
s_{11}\Delta(k),\quad |k|<k_0,\\
  1,\quad |k|>k_0;
\end{array}\right.$
\item $\operatorname{det}(\delta(k))\rightarrow 1,\ k\rightarrow\infty$.
\end{itemize}
According to the jump condition of the scalar Riemann-Hilbert problem, we obtain
\begin{gather}
(\operatorname{ln}(\operatorname{det}\delta))_+-(\operatorname{ln}(\operatorname{det}\delta))_-=\left\{\begin{array}{l}
\operatorname{ln}(s_{11}\Delta(k)),\quad |k|<k_0,\\
0,\quad |k|>k_0.
\end{array}\right.\notag
\end{gather}
According to the Plemelj formula, the $\operatorname{det}(\delta(k))$ admits the following expression
\begin{gather}
\operatorname{det}(\delta(k))=\operatorname{exp}\left[\frac{1}{2\pi i }\int_{-k_0}^{k_0}\frac{\operatorname{ln}(s_{11}(\zeta)\Delta(\zeta))}{\zeta-k}d\zeta\right].\notag
\end{gather}
Around the processing of the central block matrix $B$, we introduce a $3\times3$ matrix function
\begin{gather}
D(k)=\left(\begin{array}{cc}
  \frac{1}{\operatorname{det}(\delta(k))} &  \\
    & \delta(k)
  \end{array}\right),\notag
\end{gather}
which satisfies $\operatorname{det}(D(k))=1$ and
\begin{gather}
\left\{\begin{array}{l}
D_+(k)=D_-(k)B(k),\quad k\in (-k_0,k_0),\\
D(\infty)=I.
\end{array}\right.\notag
\end{gather}
In order to isolate the local contribution of the stationary phase points and remove the endpoint singularities generated by $D(k)$, we introduce small disks around $k=\pm k_0$ and further deform the contours inside these neighbourhoods. The division of regions around $k=\pm k_0$ is presented in Fig. \ref{fi11}, and the new jump contour is displayed in Fig. \ref{fi12} as below.

\begin{figure}[ht]
\centering
  \includegraphics[width=0.9\linewidth]{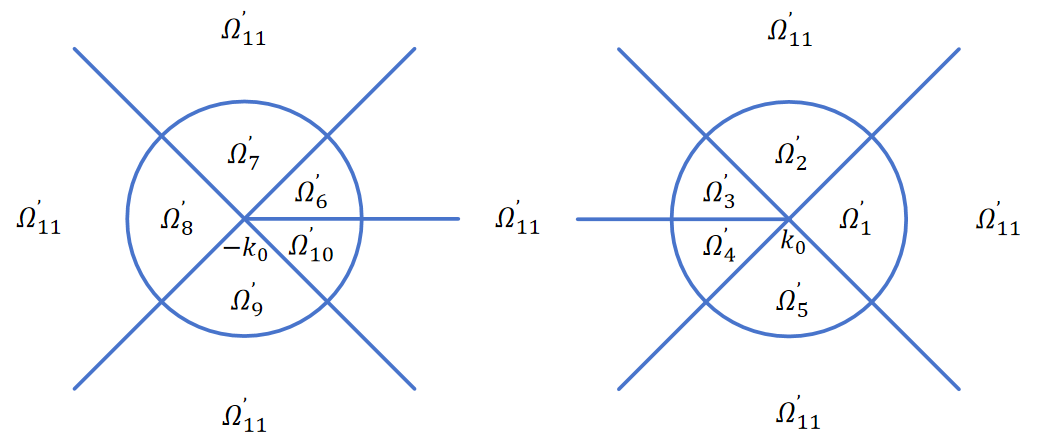}
  \caption{The division of regions around the stationary phase points.}
\label{fi11}
\end{figure}

\begin{figure}[ht]
\centering
  \includegraphics[width=0.85\linewidth]{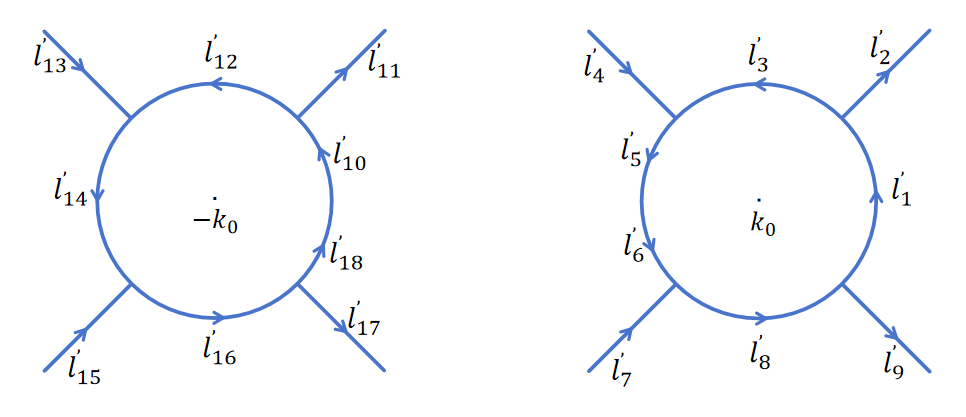}
  \caption{The jump contours around the stationary phase points of Riemann-Hilbert problem II}
\label{fi12}
\end{figure}

We define a new matrix function $M^{(2)}(x,t;k)$ on these eleven regions $\Omega^{'}_1$, $\Omega^{'}_2$, ..., $\Omega^{'}_{11}$,
\begin{gather}\label{M2}
M^{(2)}(x,t;k)=M^{(1)}(x,t;k)\left\{\begin{array}{l}
L^{-1}A,\ \ \quad k\in \Omega^{'}_1\cup\Omega^{'}_8,\\
C^{-1}B^{-1},\ \  k\in \Omega^{'}_2\cup\Omega^{'}_7,\\
B^{-1},\ \ \ \ \quad k\in \Omega^{'}_3\cup\Omega^{'}_6,\\
I,\ \ \quad \quad \quad k\in \Omega^{'}_4\cup\Omega^{'}_{10},\\
A,\ \ \ \ \quad \quad k\in \Omega^{'}_5\cup\Omega^{'}_9,\\
D^{-1},\ \quad \quad k\in \Omega^{'}_{11}.
\end{array}\right.
\end{gather}
Then, we construct the following Riemann-Hilbert problem II of the matrix function $M^{(2)}(x,t;k)$ based on the above relationship (\ref{M2}).

\noindent\textbf{Riemann-Hilbert problem II}
\begin{itemize}
\item $M^{(2)}(x,t;k)$ is analytic in $\mathbb{C}\setminus(l^{'}_1\cup ...\cup l^{'}_{18})$;
\item $M^{(2)}_{+}(x,t;k)=M^{(2)}_{-}(x,t;k)G^{(2)}(x,t;k),\ k\in l^{'}_1\cup ...\cup l^{'}_{18},$

where the jump matrix $G^{(2)}(x,t;k)$ is
\begin{gather}
G^{(2)}(x,t;k)=\left\{\begin{array}{l}
DL^{-1}A,\ \quad k\in l^{'}_1\cup l^{'}_{14},\\
DUD^{-1},\ \quad k\in l^{'}_2\cup l^{'}_{13},\\
DC^{-1}B^{-1},\  k\in l^{'}_3\cup l^{'}_{12},\\
DCD^{-1},\ \quad k\in l^{'}_4\cup l^{'}_{11},\\
DB^{-1},\ \ \ \quad k\in l^{'}_5\cup l^{'}_{10},\\
D,\quad \quad \quad \quad k\in l^{'}_6\cup l^{'}_{18},\\
DAD^{-1},\ \quad k\in l^{'}_7\cup l^{'}_{17},\\
DA,\ \ \ \quad \quad k\in l^{'}_8\cup l^{'}_{16},\\
DLD^{-1},\ \quad  k\in l^{'}_9\cup l^{'}_{15};
\end{array}\right.\notag
\end{gather}
\item $M^{(2)}(x,t;k)\rightarrow I,\ k\rightarrow\infty$.
\end{itemize}
The jump structure of $M^{(2)}(x,t;k)$ meets the requirements for solving the Riemann-Hilbert problem numerically.

\subsection{The Painlev\'{e} transition region}

In the Painlev\'{e} transition region, for $x>0$, the phase function $\psi(k)$ has no real stationary points, thus the deformation of the Riemann-Hilbert problem is similar to that in the fast decay region. We consider the case $x<0$ with $-x=O(t^{\frac{1}{3}})$, then the two real stationary points of the phase function $\psi(k)$ coalesce as $t\rightarrow\infty$. The local model used in the standard dispersive region is no longer valid, and a Painlev\'{e}-type local analysis is adopted instead. Accordingly, the deformation of Riemann-Hilbert problem in this region is simplified, with fewer jump contours than in the dispersive region.

On the area between the two stationary points $k_0$ and $-k_0$, the matrix factorization
\begin{gather}
  G(x,t;k)=A(x,t;k)B(x,t;k)C(x,t;k)\notag
\end{gather}
and double lens opening used in the dispersive region are no longer necessary. This is because, in this region, $|k|\leq\sqrt{\frac{c_1}{12}}t^{-\frac{1}{3}}$, $c_1>0$. Thus,
\begin{gather}
  |\psi(k)|=|2kx+8k^3t|\leq2|k||x|+8t|k|^3\leq\frac{4\sqrt{3}}{9}c_1^{\frac{3}{2}},\notag
\end{gather}
which means that $|\psi(k)|$ remains uniformly bounded and does not grow with $t$. As a consequence, the exponential factors in the jump matrixes no longer exhibit strong oscillatory behavior that appears in the dispersive region. Then, we restart form the original Riemann-Hilbert problem and employ only the simplified lens deformation associated with the matrix factorization $G(x,t;k)=L(x,t;k)U(x,t;k)$. Accordingly, we partition the complex-$k$ plane as Fig. \ref{fi13} and introduce the corresponding definition of transformed matrix $M^{(3)}(x,t;k)$
\begin{gather}\label{M3}
  M^{(3)}(x,t;k)=M(x,t;k)\left\{\begin{array}{l}
U^{-1},\quad k\in \Omega_1^{''}\cup\Omega_3^{''},\\
I,\ \ \ \quad k\in \Omega_2^{''}\cup\Omega_5^{''},\\
L,\ \ \ \quad k\in \Omega_4^{''}\cup\Omega_6^{''}.
\end{array}\right.
\end{gather}
And the matrix function $M^{(3)}(x,t;k)$ satisfies the following jump conditions
\begin{gather}\label{jumpM3}
  M^{(3)}_+(x,t;k)=M^{(3)}_-(x,t;k)\left\{\begin{array}{l}
U,\quad k\in l_1^{''}\cup l_3^{''},\\
G,\quad k\in l_2^{''},\\
L,\quad k\in l_4^{''}\cup l_5^{''}.
\end{array}\right.
\end{gather}

\begin{figure}[ht]
\centering
  \includegraphics[width=0.95\linewidth]{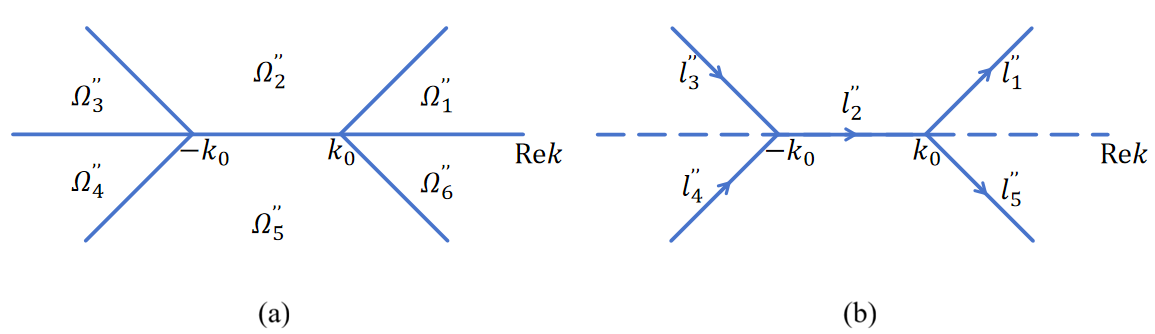}
  \caption{The division of regions and jump contours of Riemann-Hilbert problem III.}
\label{fi13}
\end{figure}

We obtain the following Riemann-Hilbert problem III of the matrix function $M^{(3)}(x,t;k)$ based on the definition (\ref{M3}) and the jump condition (\ref{jumpM3}).

\noindent\textbf{Riemann-Hilbert problem III}
\begin{itemize}
\item $M^{(3)}(x,t;k)$ is analytic in $\mathbb{C}\setminus(l^{''}_1\cup ...\cup l^{''}_5)$;
\item $M^{(3)}_{+}(x,t;k)=M^{(3)}_{-}(x,t;k)G^{(3)}(x,t;k),\ k\in l^{''}_1\cup ...\cup l^{''}_5,$

where the jump matrix $G^{(3)}(x,t;k)$ is
\begin{gather}
G^{(3)}(x,t;k)=\left\{\begin{array}{l}
U,\quad k\in l_1^{''}\cup l_3^{''},\\
G,\quad k\in l_2^{''},\\
L,\quad k\in l_4^{''}\cup l_5^{''};
\end{array}\right.\notag
\end{gather}
\item $M^{(3)}(x,t;k)\rightarrow I,\ k\rightarrow\infty$.
\end{itemize}

\subsection{The fast decay region}

In the fast decay region, the phase function $\psi(k)$ has no real stationary points, and the sign of $\operatorname{Re}\psi(k)$ in the complex-$k$ plane is shown as Fig. \ref{fi9}. We introduce $\beta$ so that $0<\beta<\sqrt{3}|k_0|$, and the points $\pm i \beta$ lie in the upper and lower half planes respectively. We use these points as anchor points to introduce two families of suitable ray contours in the upper and lower half-planes presented in Fig. \ref{fi14}, and the corresponding jump factors are continued into regions where they decay exponentially.

\begin{figure}[ht]
\centering
  \includegraphics[width=0.45\linewidth]{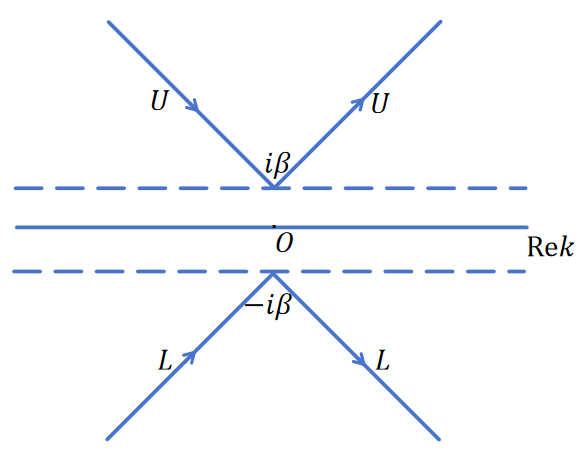}
  \caption{The jump contours and jump factors of Riemann-Hilbert problem IV.}
\label{fi14}
\end{figure}

We denote the two ray contours in the upper half-plane as $\Sigma^{(u)}$, and the two ray contours in the lower half-plane as $\Sigma^{(l)}$. Then, we construct the Riemann-Hilbert problem IV satisfied by a new matrix function $M^{(4)}(x,t;k)$.

\noindent\textbf{Riemann-Hilbert problem IV}
\begin{itemize}
\item $M^{(4)}(x,t;k)$ is analytic in $\mathbb{C}\setminus(\Sigma^{(u)}\cup \Sigma^{(l)})$;
\item $M^{(4)}_{+}(x,t;k)=M^{(4)}_{-}(x,t;k)G^{(4)}(x,t;k),\ k\in \Sigma^{(u)}\cup \Sigma^{(l)},$

where the jump matrix $G^{(4)}(x,t;k)$ is
\begin{gather}
G^{(4)}(x,t;k)=\left\{\begin{array}{l}
U,\quad k\in \Sigma^{(u)},\\
L,\quad k\in \Sigma^{(l)};
\end{array}\right.\notag
\end{gather}
\item $M^{(4)}(x,t;k)\rightarrow I,\ k\rightarrow\infty$.
\end{itemize}

\section{Numerical inverse scattering }\label{sec5}

Based on the above analysis and deformation of the oscillatory Riemann-Hilbert problem, we obtain a deformed Riemann-Hilbert problem with non-oscillatory or weakly oscillatory jump matrixes, which is suitable for numerical computation. This forms the basis for numerically studying the long-time evolution for solutions in this section. To solve the deformed Riemann-Hilbert problem effectively, we employ the Chebyshev collocation method and Olver's numerical method \citep{Olver2011,Olver2012}. First, the given Riemann-Hilbert problem exists the following form
\begin{equation}\label{RHP}
\begin{split}
&\Psi_{+}(k)=\Psi_{-}(k)G(k),\quad k\in l,\\
&\Psi(\infty)=I,
\end{split}
\end{equation}
where $l$ is the jump curve in the Riemann-Hilbert problem, and $G(k)$ is the jump matrix on jump curve. Furthermore, the jump contours $l_1$, $l_2$, ..., $l_m$ together form the jump curve $l$, and each $l_m$ is a non-self-intersecting contour. Different jump branch $l_m$ corresponds to different $k_m\in l_m$, and $k=k_1\cup k_2\cup...\cup  k_m$. Then, we consider the Cauchy transformation
\begin{gather}\label{cauchy}
  C_l(f(k))=\frac{1}{2\pi i}\int_l\frac{f(\xi)}{\xi-k}d\xi,
\end{gather}
and the left and right limits of the Cauchy transformation (\ref{cauchy})
\begin{gather}
C_l^+(f(k))=\lim\limits_{\varepsilon\to0+}\frac{1}{2\pi i }\int_l\frac{f(\xi)}{\xi-(k+i \varepsilon)}d\xi,\notag\\
C_l^-(f(k))=\lim\limits_{\varepsilon\to0+}\frac{1}{2\pi i }\int_l\frac{f(\xi)}{\xi-(k-i \varepsilon)}d\xi.\notag
\end{gather}
The Cauchy operators $C_l^+$ and $C_l^-$ satisfy $C_l^+-C_l^-=I$, so that $f(k)=C_l^+(f(k))-C_l^-(f(k))$. The solution $\Psi(k)$ of the above Riemann-Hilbert problem (\ref{RHP}) is expressed by the Cauchy transformation
\begin{gather}\label{PsiV}
\Psi(k)=I+C_lV(k)=I+\frac{1}{2\pi i}\int_l\frac{V(\xi)}{\xi-k}d\xi.
\end{gather}
Thus, the matrix function $V(k)$ satisfies the following equation
\begin{gather}\label{v}
C_l^+V(k)-C_l^-V(k)G(k)=G(k)-I,
\end{gather}
which means that solving the Riemann-Hilbert problem (\ref{RHP}) is transformed into solving the equation (\ref{v}) about $V(k)$. For conducting the Chebyshev collocation method on the equation (\ref{v}), we import $n$ Chebyshev nodes $\boldsymbol{x}=(x_1,x_2,\ldots,x_n)^{\top}$ and $n$ Chebyshev polynomials $T_i(i=0,1,\ldots,n-1)$. Furthermore, we choose appropriate M\"{o}bius transformations $M_1$, $M_2$, ..., $M_m$, which map the unit interval $[-1,1]$ to the contours $l_1$, $l_2$, ..., $l_m$ of the jump curve, that is $M_m([-1,1])=l_m$. Therefore, the matrix function $V(k)$ defined over the entire jump curve $l=l_1\cup l_2\cup...\cup  l_m$ in equation (\ref{v}) can be written as
\begin{gather}\label{Vwritten}
  V(k)=(V_1(k),\ V_2(k),\ ...,\ V_m(k))^{\top},
\end{gather}
where $V_i=V|_{l_i}$ for $k\in l_i$, and each matrix function $V_i(k)$ is a $3\times 3$ matrix, that is
\begin{gather}\label{Vmatrix}
V_i(k)=\left(\begin{array}{ccc}
V_i^{(11)}(k) & V_i^{(12)}(k) & V_i^{(13)}(k) \\
V_i^{(21)}(k) & V_i^{(22)}(k) & V_i^{(23)}(k) \\
V_i^{(31)}(k) & V_i^{(32)}(k) & V_i^{(33)}(k)
\end{array}\right),\quad i=1,\ 2,\ \ldots,\  m,
\end{gather}
and
\begin{gather}
  V_i^{(mn)}(k)=(T_0(M_i^{-1}(k)),\ T_1(M_i^{-1}(k)),\ \cdots,\ T_{n-1}(M_i^{-1}(k))) \mathcal{F}\boldsymbol{v}_i^{(mn)}.\notag
\end{gather}
The $\boldsymbol{v}_i^{(mn)}$ is the column vector composed of unknown parameters. In order to determine these unknown parameters, we substitute the expressions (\ref{Vwritten}) and (\ref{Vmatrix}) into the equation (\ref{v}), and construct a $9\cdot m\cdot n\times9\cdot m\cdot n$ linear system as
\begin{gather}\label{linearsystem}
[C_+-\operatorname{rdiag}(G)C_-]V=\operatorname{diag}(G)-I,
\end{gather}
where $C_+$ and $C_-$ are the constructed Cauchy matrixes. Through solving the linear system (\ref{linearsystem}) numerically, we obtain the unknown parameters $\boldsymbol{v}_i^{(mn)}$ and derive the matrix function $V(k)$.

We consider that the expression of $\Psi(k)$ is
\begin{gather}
\Psi(k)=I+C_lV(k)=I+C_{l_1}V_1(k)+C_{l_2}V_2(k)+\cdots+C_{l_m}V_m(k),\notag
\end{gather}
then we obtian
\begin{gather}
\begin{aligned}
\lim\limits_{k\to\infty}k\Psi^{(12)}(k)&=\lim\limits_{k\to\infty}\frac{1}{2\pi i }\int_l\frac{k}{\xi-k}V^{(12)}(\xi)d\xi\\
&=-\frac{1}{2\pi i }\left[\int_{l_1}V_1^{(12)}(\xi)d\xi+...+\int_{l_m}V_m^{(12)}(\xi)d\xi\right],
\end{aligned}
\end{gather}
\begin{gather}
\begin{aligned}
\lim\limits_{k\to\infty}k\Psi^{(13)}(k)&=\lim\limits_{k\to\infty}\frac{1}{2\pi i }\int_l\frac{k}{\xi-k}V^{(13)}(\xi)d\xi\\
&=-\frac{1}{2\pi i }\left[\int_{l_1}V_1^{(13)}(\xi)d\xi+...+\int_{l_m}V_m^{(13)}(\xi)d\xi\right],
\end{aligned}
\end{gather}
thus, we provide the overall framework for numerically solving the Riemann-Hilbert problem associated with the coupled mKdV equation (\ref{cmKdV}) by means of the NIST method. In addition to implementing the NIST approach, it is also natural to compare this method with the traditional numerical method, such as the Fourier spectral method (FSM). The FSM provides an effective way to solve the coupled mKdV equation (\ref{cmKdV}) on a truncated interval $[-L,\ L]$ under the periodic boundary conditions
\begin{gather}
  u(-L,t)=u(L,t),\quad v(-L,t)=v(L,t).\notag
\end{gather}
For the case $\alpha=1$, we take the Fourier transform of the coupled mKdV equation (\ref{cmKdV})
\begin{gather}\label{Fourier}
\hat{u}_t=i k^3\hat{u}+\operatorname{F}[N_u],\quad \hat{v}_t=i k^3\hat{v}+\operatorname{F}[N_v],
\end{gather}
where
\begin{gather}
N_u=3(u_xv^2+uvv_x+2u^2u_x),\quad N_v=3(u^2v_x+uu_xv+2v^2v_x),\notag\\
u_x=\operatorname{F}^{-1}(i k\hat{u}),\quad v_x=\operatorname{F}^{-1}(i k\hat{v}).\notag
\end{gather}
We multiply $e^{-i k^3t}$ to these expressions in (\ref{Fourier}), and then the following equations are obtained as
\begin{gather}\label{Fourier1}
(e^{-i k^3t}\hat{u})_t=e^{-i k^3t}\operatorname{F}[N_u],\quad (e^{-i k^3t}\hat{v})_t=e^{-i k^3t}\operatorname{F}[N_v],
\end{gather}
where $\hat{u}$ and $\hat{v}$ represent the Fourier transform of functions $u$ and $v$ respectively, $\operatorname{F}[\ \cdot\ ]$ denotes the Fourier transform, and $\operatorname{F}^{-1}[\ \cdot\ ]$ denotes the inverse Fourier transform. The explicit adaptive Runge-Kutta method, such as ode45, is employed to solve the equation (\ref{Fourier1}), and the solutions $u$ and $v$ can be derived with the inverse Fourier transform of the $\hat{u}$ and $\hat{v}$. For the cubic nonlinear terms, the 2/3-rule dealiasing can be further employed to reduce the aliasing error.

The previously constructed NIST framework is first applied to the coupled mKdV equation (\ref{cmKdV}) with the initial profiles $u_0=0.5\operatorname{sech}(x)$ and $v_0=1.5\operatorname{exp}(-x^2)$. For short times, no contour deformation is required in the associated Riemann-Hilbert problem, and the corresponding numerical solutions $u$ and $v$ are presented in Fig. \ref{fi15} (a) and (b), respectively. The results clearly show the spatiotemporal evolution of the solutions in the short-time regime. In particular, compared with existing numerical studies for scalar defocusing mKdV equations \citep{Trogdon2012,Lin2023,Wang2022}, the present computation reveals nontrivial wave patterns for the coupled mKdV equation obtained via the NIST framework. To assess the accuracy of the NIST approach, we compare the NIST results with those obtained from the FSM. The comparisons at $t=0.1$, $t=0.3$ and $t=0.5$ are displayed in Figs. \ref{fi16}-\ref{fi18}, respectively, where the blue solid curves correspond to the FSM and the red dotted curves correspond to the NIST. The two sets of results agree very well in the short-time regime, and the absolute errors remain below $10^{-5}$. This indicates that both methods are reliable for short-time computations.

\begin{figure}[ht]
\centering
  \includegraphics[width=0.85\linewidth]{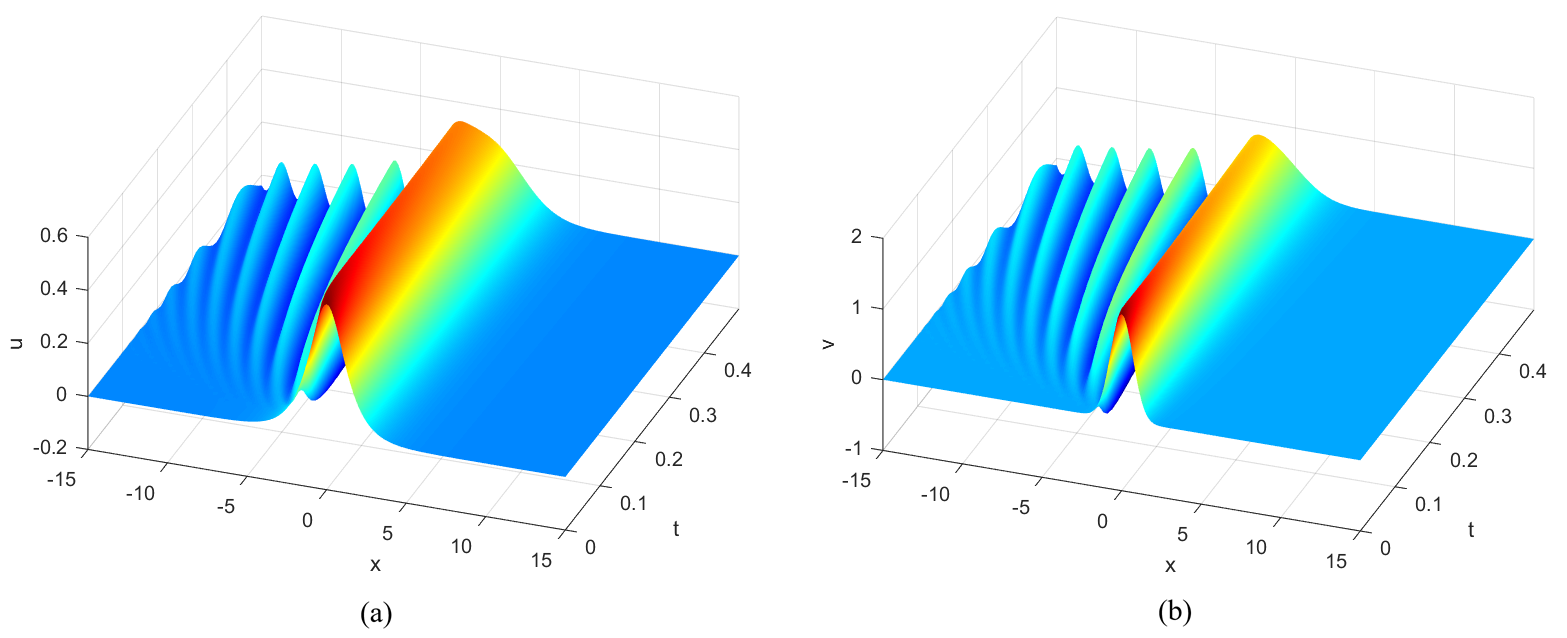}
  \caption{The spatiotemporal evolution of $u$ and $v$ computed by the NIST for initial profiles $u_0=0.5\operatorname{sech}(x)$ and $v_0=1.5\operatorname{exp}(-x^2)$ in the short time.}
\label{fi15}
\end{figure}

\begin{figure}[ht]
\centering
  \includegraphics[width=0.8\linewidth]{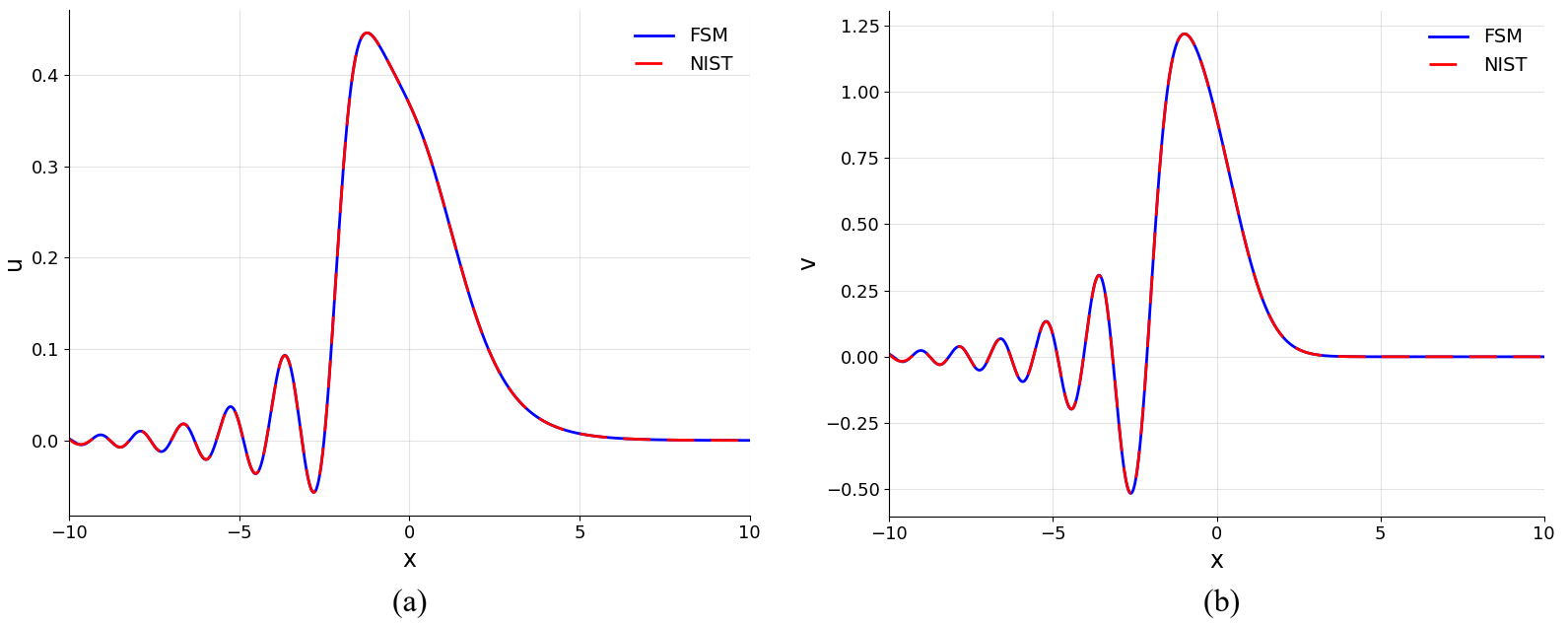}
  \caption{Comparison of the numerical results computed by NIST and FSM at $t=0.1$: (a) $u(x,t)$; (b) $v(x,t)$.}
\label{fi16}
\end{figure}

\begin{figure}[ht]
\centering
  \includegraphics[width=0.8\linewidth]{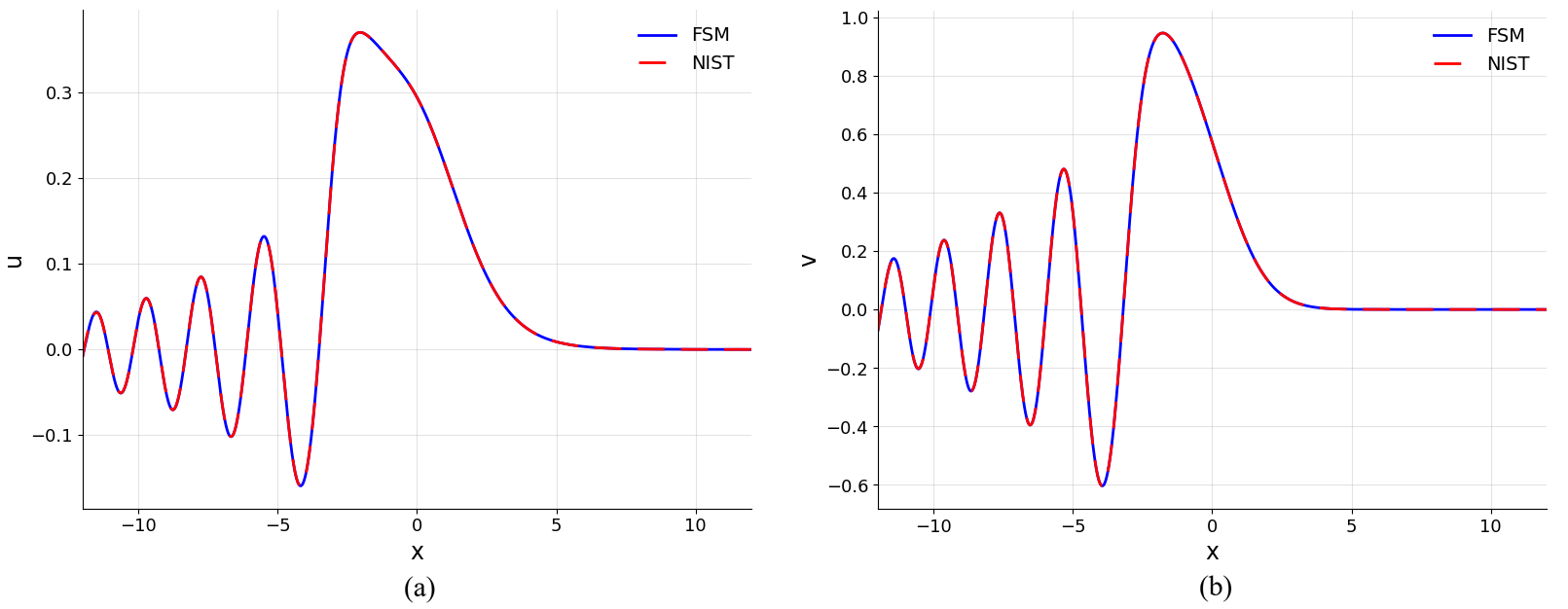}
  \caption{Comparison of the numerical results computed by NIST and FSM at $t=0.3$: (a) $u(x,t)$; (b) $v(x,t)$.}
\label{fi17}
\end{figure}

\begin{figure}[ht]
\centering
  \includegraphics[width=0.8\linewidth]{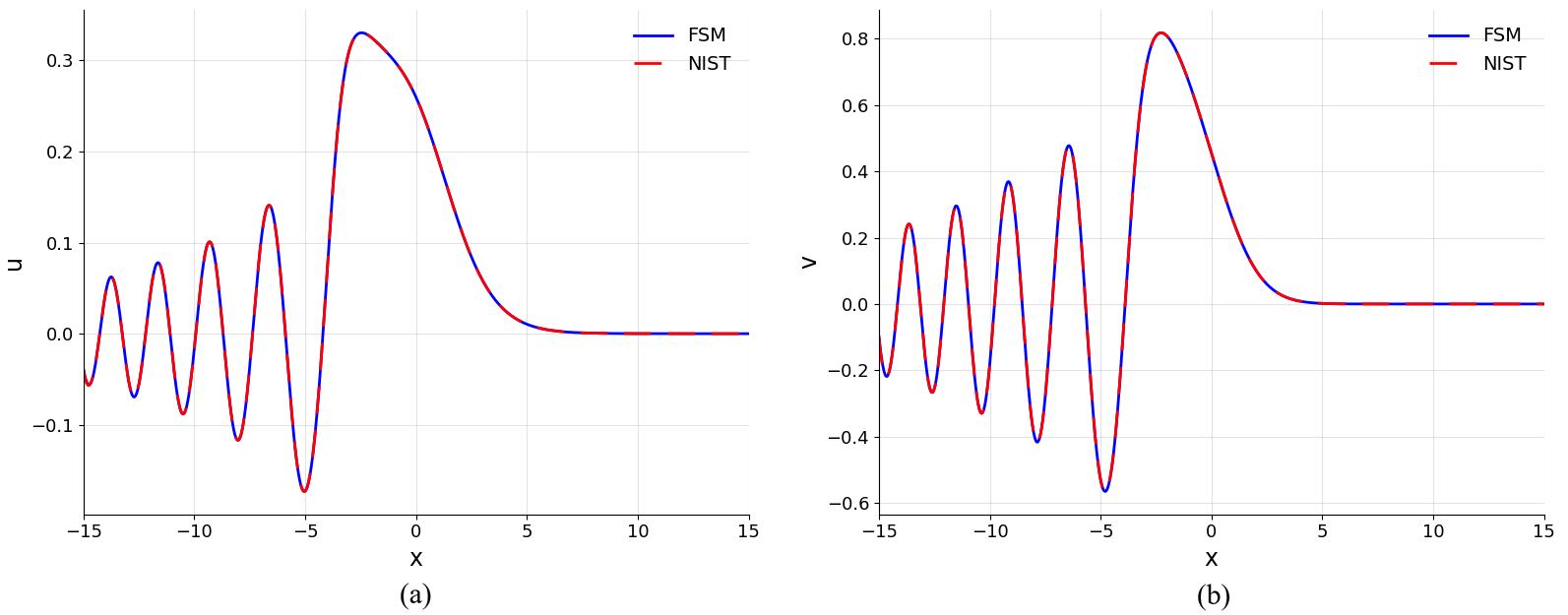}
  \caption{Comparison of the numerical results computed by NIST and FSM at $t=0.5$: (a) $u(x,t)$; (b) $v(x,t)$.}
\label{fi18}
\end{figure}

To further illustrate the advantage of the NIST over the FSM as time increases, we compare the numerical results of these two methods at $t=10$, $t=15$ and $t=20$, as shown in Fig. \ref{fi19}. In the figure, the blue curves represent the calculated results of the FSM and the red dotted curves denote the NIST. In terms of the profiles of $u$ and $v$, the NIST results remain smooth and stable, whereas the FSM results exhibit increasingly pronounced oscillations. It is observed that the oscillations of the FSM solution become increasingly severe as time grows, especially when the waves approach the boundaries of the truncated computational interval. Hence, although the FSM performs well for short times, it is no longer reliable for long-time evolution. In contrast, the NIST remains smooth and stable numerical results, showing its effectiveness for numerical computation of the coupled mKdV equation (\ref{cmKdV}) when the time is relatively large.

\begin{figure}[ht]
\centering
  \includegraphics[width=0.99\linewidth]{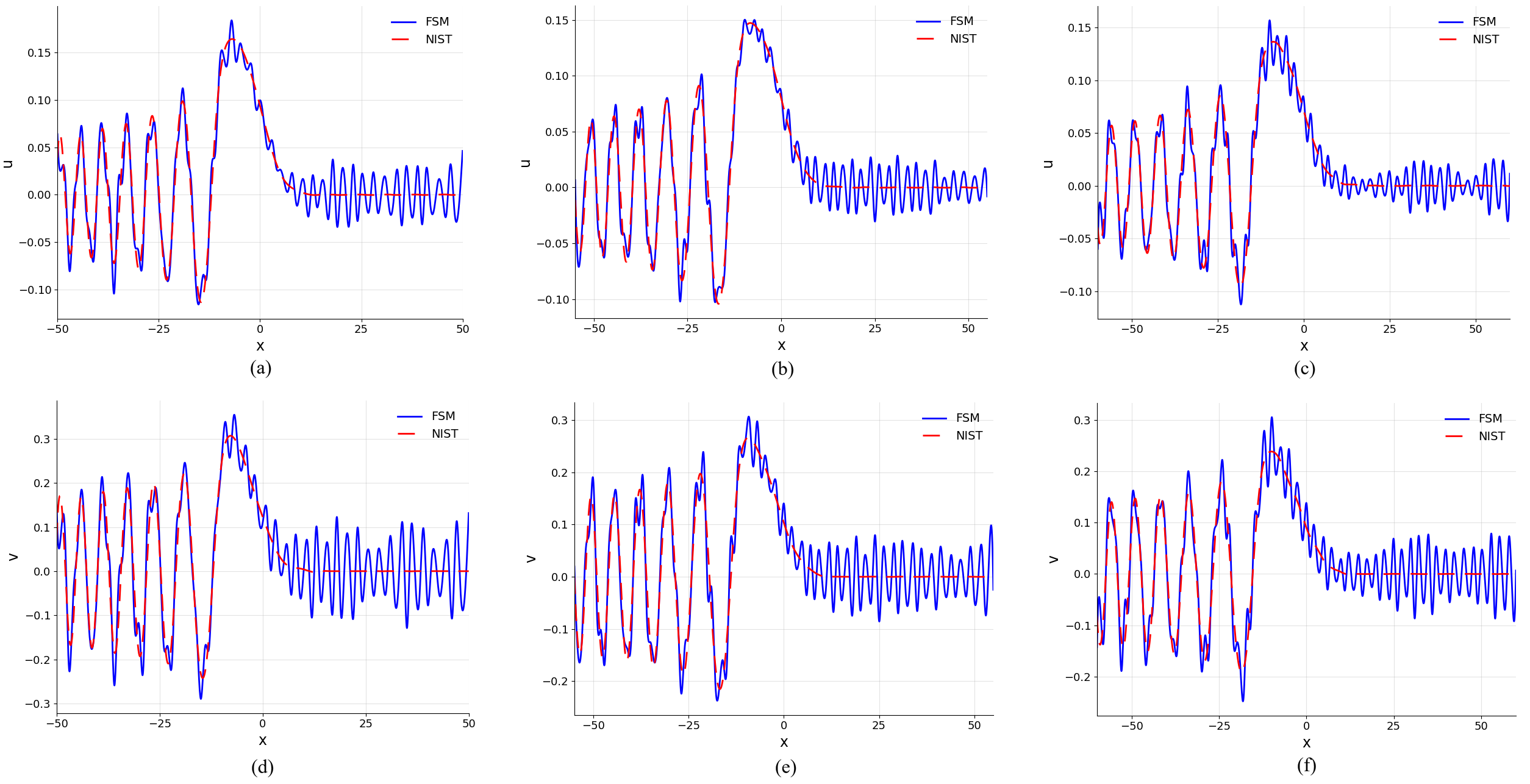}
  \caption{Comparison of the numerical solutions computed by NIST and FSM. (a)-(c): $u(x,t)$ at $t=10$, $t=15$ and $t=20$; (d)-(f): $v(x,t)$ at $t=10$, $t=15$ and $t=20$.}
\label{fi19}
\end{figure}

The NIST can also be effectively applied to compute the long-time evolution of the coupled mKdV equation (\ref{cmKdV}). The numerical solutions $u$ and $v$ at $t=50$, $t=100$, and $t=150$ are shown in Figs. \ref{fi20}-\ref{fi22}. The computational cost of the NIST does not increase with time, and the method remains stable in the long-time regime. As far as we know, traditional numerical methods, such as the FSM, are difficult to calculate the long-term evolution of solutions, however these numerical results indicate that the NIST provides an effective tool for studying with the long-time period.

Moreover, the computed profiles further reveal clear qualitative features of the solutions. As time evolves, the principal peak propagates to the left, namely in the negative $x$-direction, while its amplitude gradually decreases.  Meanwhile, a pronounced oscillatory tail develops on the left side of the main peak, and the local oscillation amplitude gradually decreases away from the peak. Thus, the NIST is able not only to reliably compute solutions for large times, but also to successfully capture their essential asymptotic structure.

\begin{figure}[ht]
\centering
  \includegraphics[width=0.85\linewidth]{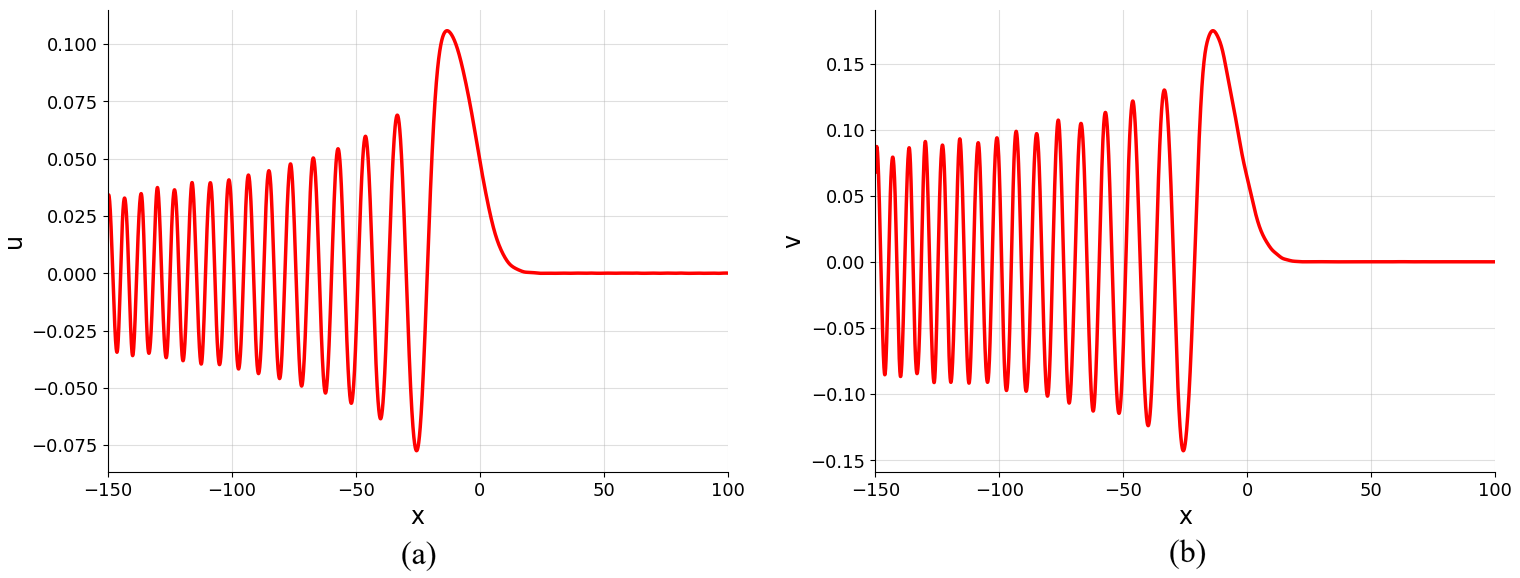}
  \caption{Long-time numerical profiles computed by the NIST method at $t=50$: (a) $u(x,t)$; (b) $v(x,t)$.}
\label{fi20}
\end{figure}

\begin{figure}[ht]
\centering
  \includegraphics[width=0.85\linewidth]{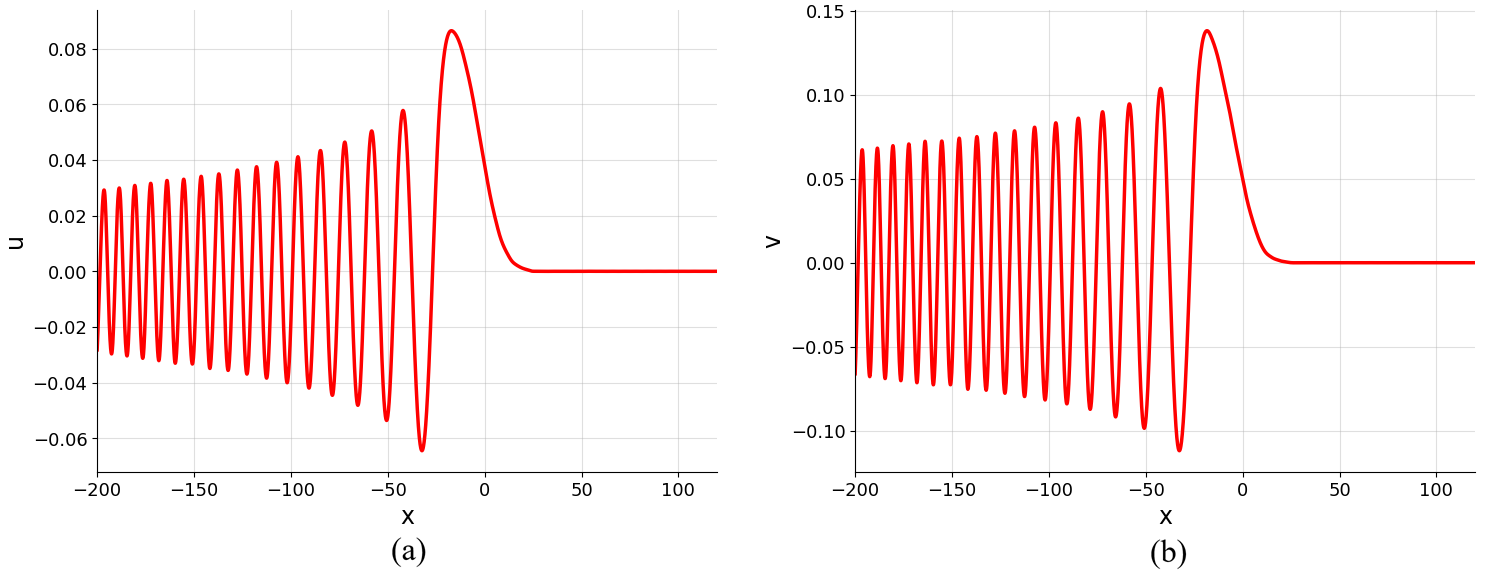}
  \caption{Long-time numerical profiles computed by the NIST method at $t=100$: (a) $u(x,t)$; (b) $v(x,t)$.}
\label{fi21}
\end{figure}

\begin{figure}[ht]
\centering
  \includegraphics[width=0.85\linewidth]{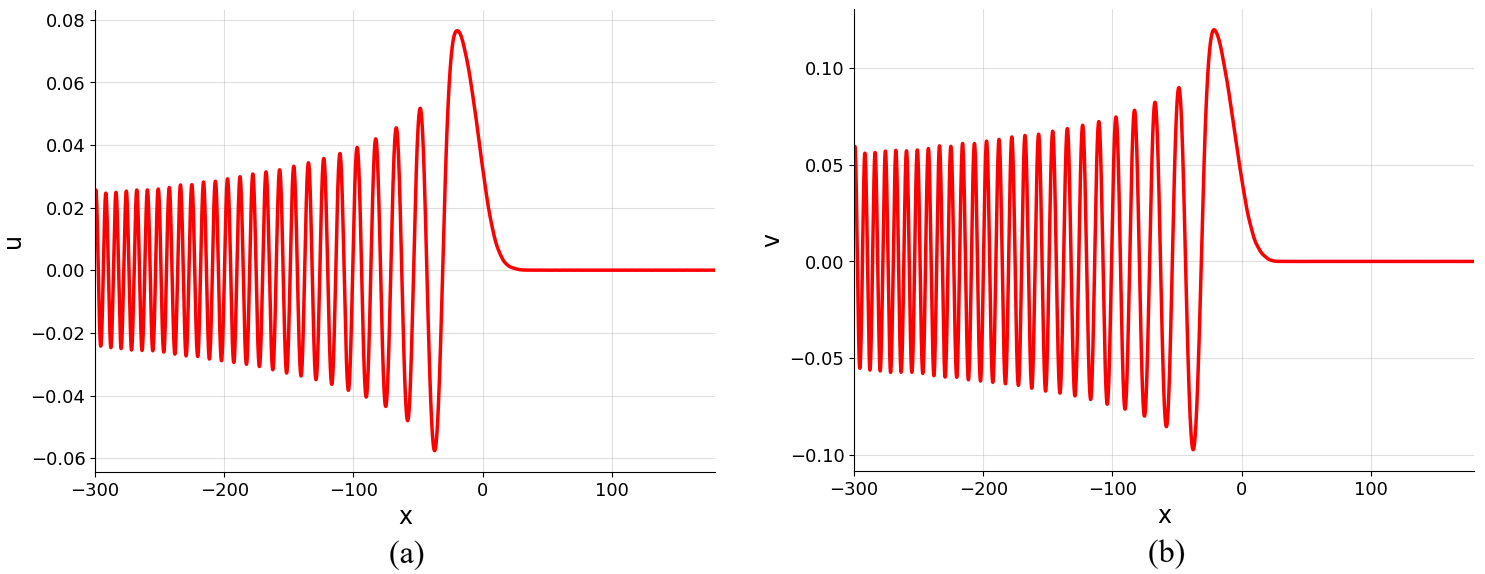}
  \caption{Long-time numerical profiles computed by the NIST method at $t=150$: (a) $u(x,t)$; (b) $v(x,t)$.}
\label{fi22}
\end{figure}


\section{Conclusions}\label{sec6}

In this paper, we develop the NIST framework for the coupled mKdV equation with initial data in the Schwartz space $S(\mathbb{R})$. Unlike traditional numerical methods, the NIST computes the solutions directly at prescribed spatial and temporal points through the associated Riemann-Hilbert problem, without relying on time-stepping procedures. Before implementing the NIST, two preparatory steps are needed. First, based on the Lax pair, we carry out the spectral analysis and construct the Riemann-Hilbert problem with a $3\times3$ jump matrix. For the direct problem, the scattering data are computed with high accuracy by using the Chebyshev collocation method together with appropriate mapping functions. The convergence analysis shows that the resulting numerical scheme is stable and convergent. Second, to deal with the oscillatory Riemann-Hilbert problem, we introduce factorizations of the jump matrix and contour deformations based on the Deift-Zhou nonlinear steepest descent method. In particular, the $(x,t)$-plane is divided into three regions to construct suitable contour deformations. On this basis, the NIST is implemented to solve the Riemann-Hilbert problem numerically, and its performance is compared with that of the FSM. The numerical results demonstrate that the NIST is effective and accurate for the computation of the coupled mKdV equation, especially in the long-time regime. Overall, the present work extends the NIST framework to the coupled mKdV system and provides an effective numerical approach for studying its solution dynamics through the associated Riemann-Hilbert problem.
\\

\noindent\textbf{Author Contributions} All authors contributed to the research and preparation of this work.\\

\noindent\textbf{Funding} This work is supported by National Natural Science Foundation of China (Nos. 12575002 and 12235007), Science and Technology Commission of Shanghai Municipality (Nos. 21JC1402500 and 22DZ2229014).\\

\noindent\textbf{Data Availability} Data will be made available on request.

\section*{Declarations}

\noindent\textbf{Conflict of interest} The authors declare no conflict of interest.


\bibliography{sn-bibliography}

@article{Wadati1973,
  author		= "Wadati, M.",
  title			= "The modified {K}orteweg-de {V}ries equation",
  journal		= "J. Phys. Soc. Jpn.",
  volume		= "34",
  number		= "5",
  pages			= "1289--1296",
  year			= "1973"
}

@article{Miura1968,
  author		= "Miura, R. M.",
  title			= "Korteweg-de {V}ries equation and generalizations. {I}. {A} remarkable explicit nonlinear transformation",
  journal		= "J. Math. Phys.",
  volume		= "9",
  number		= "8",
  pages			= "1202--1204",
  year			= "1968"
}

@article{Grimshaw1997,
  author		= "Grimshaw, R. and Pelinovsky, E. and Talipova, T.",
  title			= "The modified {K}orteweg-de {V}ries equation in the theory of large-amplitude internal waves",
  journal		= "Nonlinear Process. Geophys.",
  volume		= "4",
  number		= "4",
  pages			= "237--250",
  year			= "1997"
}

@article{Mel’nikov1997,
  author		= "Mel’nikov, I. V. and Mihalache, D. and Moldoveanu, F. and Panoiu, N. C.",
  title			= "Quasiadiabatic following of femtosecond optical pulses in a weakly excited semiconductor",
  journal		= "Phys. Rev. A",
  volume		= "56",
  number		= "2",
  pages			= "1569",
  year			= "1997"
}

@article{Leblond2018,
  author		= "Leblond, H. and Mihalache, D.",
  title			= "Ultrashort spatiotemporal optical solitons in waveguide arrays: the effect of combined linear and nonlinear couplings",
  journal		= "J. Phys. A: Math. Theor.",
  volume		= "51",
  number		= "43",
  pages			= "435202",
  year			= "2018"
}

@article{Khater1998,
  author		= "Khater, A. H. and El-Kalaawy, O. H. and Callebaut, D. K.",
  title			= {B\"{a}cklund transformations and exact solutions for {A}lfv\'{e}n solitons in a relativistic electron-positron plasma},
  journal		= "Phys. Scr.",
  volume		= "58",
  number		= "6",
  pages			= "545--548",
  year			= "1998"
}

@article{Ono1992,
  author		= "Ono, H.",
  title			= "Soliton fission in anharmonic lattices with reflectionless inhomogeneity",
  journal		= "J. Phys. Soc. Jpn.",
  volume		= "61",
  number		= "12",
  pages			= "4336--4343",
  year			= "1992"
}

@article{Helal2002,
  author		= "Helal, M. A.",
  title			= "Soliton solution of some nonlinear partial differential equations and its applications in fluid mechanics",
  journal		= "Chaos Solitons Fractals",
  volume		= "13",
  number		= "9",
  pages			= "1917--1929",
  year			= "2002"
}

@article{Yajima1975,
  author		= "Yajima, N. and Oikawa, M.",
  title			= "A class of exactly solvable nonlinear evolution equations",
  journal		= "Prog. Theor. Phys.",
  volume		= "54",
  number		= "5",
  pages			= "1576--1577",
  year			= "1975"
}

@article{Athorne1987,
  author		= "Athorne, C. and Fordy, A.",
  title			= "Generalised {K}d{V} and {M}{K}d{V} equations associated with symmetric spaces",
  journal		= "J. Phys. A: Math. Gen.",
  volume		= "20",
  number		= "6",
  pages			= "1377--1386",
  year			= "1987"
}

@article{Hirota1997,
  author		= "Hirota, R.",
  title			= {``Molecule solutions" of coupled modified {K}d{V} equations},
  journal		= "J. Phys. Soc. Jpn.",
  volume		= "66",
  number		= "9",
  pages			= "2530--2532",
  year			= "1997"
}

@article{Iwao1997,
  author		= "Iwao, M. and Hirota, R.",
  title			= "Soliton solutions of a coupled modified {K}d{V} equations",
  journal		= "J. Phys. Soc. Jpn.",
  volume		= "66",
  number		= "3",
  pages			= "577--588",
  year			= "1997"
}

@article{Tsuchida1998,
  author		= "Tsuchida, T. and Wadati, M.",
  title			= "The coupled modified {K}orteweg-de {V}ries equations",
  journal		= " J. Phys. Soc. Jpn.",
  volume		= "67",
  number		= "4",
  pages			= "1175--1187",
  year			= "1998"
}

@article{Wu2017,
  author		= "Wu, J. and Geng, X.",
  title			= "Inverse scattering transform and soliton classification of the coupled modified {K}orteweg-de {V}ries equation",
  journal		= "Commun. Nonlinear Sci. Numer. Simul.",
  volume		= "53",
  pages			= "83--93",
  year			= "2017"
}

@article{Ma2019,
  author		= "Ma, W. X.",
  title			= "The inverse scattering transform and soliton solutions of a combined modified {K}orteweg-de {V}ries equation",
  journal		= "J. Math. Anal. Appl.",
  volume		= "471",
  number		= "1--2",
  pages			= "796--811",
  year			= "2019"
}

@article{Xue2015,
  author		= "Xue, B. and Li, F. and Yang, G.",
  title			= "Explicit solutions and conservation laws of the coupled modified {K}orteweg-de {V}ries equation",
  journal		= "Phys. Scr.",
  volume		= "90",
  number		= "8",
  pages			= "085204",
  year			= "2015"
}

@article{Hu2009,
  author		= "Hu, H. C.",
  title			= "Analytical positon, negaton and complexiton solutions for the coupled modified {K}d{V} system",
  journal		= "J. Phys. A: Math. Theor.",
  volume		= "42",
  number		= "18",
  pages			= "185207",
  year			= "2009"
}

@article{Tian2018,
  author		= "Tian, S. F.",
  title			= "Initial-boundary value problems for the coupled modified {K}orteweg-de {V}ries equation on the interval",
  journal		= "Commun. Pure Appl. Anal",
  volume		= "17",
  number		= "3",
  pages			= "923--957",
  year			= "2018"
}

@article{Ma2018,
  author		= "Ma, W. X.",
  title			= "Riemann-{H}ilbert problems and {N}-soliton solutions for a coupled m{K}d{V} system",
  journal		= "J. Geom. Phys.",
  volume		= "132",
  pages			= "45--54",
  year			= "2018"
}

@article{Xu2023,
  author		= "Xu, S.",
  title			= "A coupled complex m{K}d{V} equation and its {N}-soliton solutions via the {R}iemann-{H}ilbert approach",
  journal		= "Bound. Value Probl.",
  volume		= "2023",
  number		= "1",
  pages			= "83",
  year			= "2023"
}

@article{Fan2001,
  author		= "Fan, E.",
  title			= "Soliton solutions for a generalized {H}irota-{S}atsuma coupled {K}d{V} equation and a coupled {M}{K}d{V} equation",
  journal		= "Phys. Lett. A",
  volume		= "282",
  number		= "1--2",
  pages			= "18--22",
  year			= "2001"
}

@article{Geng2014,
  author		= "Geng, X. and Zhai, Y. and Dai, H. H.",
  title			= "Algebro-geometric solutions of the coupled modified {K}orteweg-de {V}ries hierarchy",
  journal		= "Adv. Math.",
  volume		= "263",
  pages			= "123--153",
  year			= "2014"
}

@misc{Ling2025,
  author		= "Ling, L. and Su, H.",
  title			= "Nonlinear stability of vector multi-solitons in coupled {N}{L}{S} and modified {K}d{V} equations",
  year			= "2025",
  note			= "Preprint at \url{https://arxiv.org/abs/2510.12129}"
}

@article{Gardner1967,
  author		= "Gardner, C. S. and Greene, J. M. and Kruskal, M. D. and Miura, R. M.",
  title			= "Method for solving {K}orteweg-de {V}ries equation",
  journal		= "Phys. Rev. Lett.",
  volume		= "19",
  pages			= "1095--1097",
  year			= "1967"
}

@article{Ablowitz1974,
  author		= "Ablowitz, M. J. and Kaup, D. J. and Newell, A. C. and Segur, H.",
  title			= "The {I}nverse {S}cattering {T}ransform-{F}ourier {A}nalysis for {N}onlinear {P}roblems",
  journal		= "Stud. Appl. Math.",
  volume		= "53",
  number		= "4",
  pages			= "249--315",
  year			= "1974"
}

@book{Ablowitz1981,
  author		= "Ablowitz, M. J. and Segur, H.",
  title			= "Solitons and the inverse scattering transform",
  address		= "Philadelphia",
  publisher		= "S{I}{A}{M}",
  year			= "1981"
}

@book{Novikov1984,
  author		= "Novikov, S. P. and Manakov, S. V. and Pitaevskii, L. P. and Zakharov, V. E.",
  title			= "Theory of {S}olitons: {T}he {I}nverse {S}cattering {M}ethod",
  address		= "New {Y}ork",
  publisher		= "Consul-tants {B}ureau",
  year			= "1984"
}

@book{Yang2010,
  author		= "Yang, J. K.",
  title			= "Nonlinear {W}aves in {I}ntegrable and {N}onintegrable {S}ystems",
  address		= "Philadelphia",
  publisher		= "S{I}{A}{M}",
  year			= "2010"
}

@article{Deift1993,
  author		= "Deift, P. and Zhou, X.",
  title			= "A steepest descent method for oscillatory {R}iemann-{H}ilbert problems. {A}symptotics for the {M}{K}d{V} equation",
  journal		= "Ann. of Math.",
  volume		= "137",
  number		= "2",
  pages			= "295--368",
  year			= "1993"
}

@article{Deift1994,
  author		= "Deift, P. and Zhou, X.",
  title			= "Long-time asymptotics for integrable systems. {H}igher order theory",
  journal		= "Commun. Math. Phys.",
  volume		= "165",
  pages			= "175--191",
  year			= "1994"
}

@article{Grunert2009,
  author		= "Grunert, K. and Teschl, G.",
  title			= "Long-{T}ime {A}symptotics for the {K}orteweg-de {V}ries {E}quation via {N}onlinear {S}teepest {D}escent",
  journal		= "Math. Phys. Anal. Geom.",
  volume		= "12",
  pages			= "287--324",
  year			= "2009"
}

@article{Kamvissis1996,
  author		= "Kamvissis, S.",
  title			= {Long time behavior for the focusing nonlinear {S}chr\"{o}dinger equation with real spectral singularities},
  journal		= "Commun. Math. Phys.",
  volume		= "180",
  pages			= "325--341",
  year			= "1996"
}

@article{Buckingham2007,
  author		= "Buckingham, R. and Venakides, S.",
  title			= {Long-time asymptotics of the nonlinear {S}chr\"{o}dinger equation shock problem},
  journal		= "Commun. Pure Appl. Math.",
  volume		= "60",
  pages			= "1349--1414",
  year			= "2007"
}

@article{Xu2013,
  author		= "Xu, J. and Fan, E. and Chen, Y.",
  title			= {Long-time {A}symptotic for the {D}erivative {N}onlinear {S}chr\"{o}dinger {E}quation with {S}tep-like {I}nitial {V}alue},
  journal		= "Math. Phys. Anal. Geom.",
  volume		= "16",
  pages			= "253--288",
  year			= "2013"
}

@article{Cheng1999,
  author		= "Cheng, P. J. and Venakides, S. and Zhou, X.",
  title			= "Long-time asymptotics for the pure radiation solution of the sine-{G}ordon equation",
  journal		= "Comm. Partial Differential Equations.",
  volume		= "24",
  pages			= "1195--1262",
  year			= "1999"
}

@article{Monvel2009,
  author		= "De Monvel, A. B. and Kostenko, A. and Shepelsky, D. and Teschl, G.",
  title			= "Long-time asymptotics for the {C}amassa-{H}olm equation",
  journal		= "SIAM J. Math. Anal.",
  volume		= "41",
  pages			= "1559--1588",
  year			= "2009"
}

@article{Wang2018,
  author		= "Wang, D. S. and Wang, X. L.",
  title			= "Long-time asymptotics and the bright {N}-soliton solutions of the {K}undu-{E}ckhaus equation via the {R}iemann-{H}ilbert approach",
  journal		= "Nonlinear Anal. Real World Appl.",
  volume		= "41",
  pages			= "334--361",
  year			= "2018"
}

@article{Geng2019,
  author		= "Geng, X. and Chen, M. and Wang, K.",
  title			= "Long-time asymptotics of the coupled modified {K}orteweg-de {V}ries equation",
  journal		= "J. Geom. Phys.",
  volume		= "142",
  pages			= "151--167",
  year			= "2019"
}

@article{Ma2019-2,
  author		= "Ma, W. X.",
  title			= "Long-time asymptotics of a three-component coupled m{K}d{V} system",
  journal		= "Mathematics",
  volume		= "7",
  number		= "7",
  pages			= "573",
  year			= "2019"
}

@article{Trogdon2012,
  author		= "Trogdon, T. and Olver, S. and Deconinck, B.",
  title			= "Numerical inverse scattering for the {K}orteweg-de {V}ries and modified {K}orteweg-de {V}ries equations",
  journal		= "Phys. D",
  volume		= "241",
  pages			= "1003--1025",
  year			= "2012"
}

@article{Olver2011,
  author		= "Olver, S.",
  title			= {Numerical {S}olution of {R}iemann-{H}ilbert {P}roblems: {P}ainlev\'{e} {I}{I}},
  journal		= "Found. Comput. Math.",
  volume		= "11",
  pages			= "153--179",
  year			= "2011"
}

@article{Olver2012,
  author		= "Olver, S.",
  title			= "A general framework for solving {R}iemann-{H}ilbert problems numerically",
  journal		= "Numer. Math.",
  volume		= "122",
  pages			= "305--340",
  year			= "2012"
}

@book{Trogdon2016,
  author		= "Trogdon, T. and Olver, S.",
  title			= "Riemann-{H}ilbert problems, their numerical solution, and the computation of nonlinear special functions",
  address		= "Philadelphia",
  publisher		= "S{I}{A}{M}",
  year			= "2016"
}

@article{Trogdon2013,
  author		= "Trogdon, T. and Deconinck, B.",
  title			= "A {R}iemann-{H}ilbert problem for the finite-genus solutions of the {K}d{V} equation and its numerical solution",
  journal		= "Phys. D",
  volume		= "251",
  pages			= "1--18",
  year			= "2013"
}

@article{Bilman2020,
  author		= "Bilman, D. and Trogdon, T.",
  title			= "On numerical inverse scattering for the {K}orteweg-de {V}ries equation with discontinuous step-like data",
  journal		= "Nonlinearity",
  volume		= "33",
  pages			= "2211--2269",
  year			= "2020"
}

@article{Trogdon2013-2,
  author		= "Trogdon, T. and Olver, S.",
  title			= {Numerical inverse scattering for the focusing and defocusing nonlinear {S}chr\"{o}dinger equations},
  journal		= "Proc. A",
  volume		= "469",
  pages			= "20120330",
  year			= "2013"
}

@article{Gkogkou2026,
  author		= "Gkogkou, A. and Prinari, B. and Trogdon, T.",
  title			= {Numerical inverse scattering transform for the defocusing nonlinear {S}chr\"{o}dinger equation with box-type initial conditions on a nonzero background},
  journal		= "Journal of Nonlinear Waves",
  volume		= "2",
  pages			= "1--55",
  year			= "2026"
}

@article{Cui2024,
  author		= "Cui, S. and Wang, Z.",
  title			= {Numerical inverse scattering transform for the derivative nonlinear {S}chr\"{o}dinger equation},
  journal		= "Nonlinearity",
  volume		= "37",
  pages			= "105015",
  year			= "2024"
}

@article{Zhang2025,
  author		= "Zhang, W. X. and Chen, Y.",
  title			= {Numerical inverse scattering transform for the coupled nonlinear {S}chr\"{o}dinger equation},
  journal		= "Chaos Solitons Fractals",
  volume		= "201",
  pages			= "117185",
  year			= "2025"
}

@article{Cui2023,
  author		= "Cui, S. and Wang, Z.",
  title			= "Numerical inverse scattering transform for the focusing and defocusing {K}undu-{E}ckhaus equations",
  journal		= "Phys. D",
  volume		= "454",
  pages			= "133838",
  year			= "2023"
}

@article{Deconinck2019,
  author		= "Deconinck, B. and Trogdon, T. and Yang, X.",
  title			= "Numerical inverse scattering for the sine-{G}ordon equation",
  journal		= "Phys. D",
  volume		= "399",
  pages			= "159--172",
  year			= "2019"
}

@article{Bilman2017,
  author		= "Bilman, D. and Trogdon, T.",
  title			= "Numerical {I}nverse {S}cattering for the {T}oda {L}attice",
  journal		= "Commun. Math. Phys.",
  volume		= "352",
  pages			= "805--879",
  year			= "2017"
}

@article{Akyuz1999,
  author		= {Aky\"{u}z, A. and Sezer, M.},
  title			= "A {C}hebyshev collocation method for the solution of linear integro-differential equations",
  journal		= "Int. J. Comput. Math.",
  volume		= "72",
  number		= "4",
  pages			= "491--507",
  year			= "1999"
}

@article{Cui2024-2,
  author		= "Cui, S. and Wang, Z.",
  title			= "Efficient method to calculate the eigenvalues of the {Z}akharov-{S}habat system",
  journal		= "Chin. Phys. B",
  volume		= "33",
  pages			= "010201",
  year			= "2024"
}

@article{Lin2023,
  author		= "Lin, S. and Chen, Y.",
  title			= "Physics-informed neural network methods based on {M}iura transformations and discovery of new localized wave solutions",
  journal		= "Phys. D",
  volume		= "445",
  pages			= "133629",
  year			= "2023"
}

@article{Wang2022,
  author		= "Wang, D. S. and Xu, L. and Xuan, Z.",
  title			= "The complete classification of solutions to the {R}iemann problem of the defocusing complex modified {K}d{V} equation",
  journal		= "J. Nonlinear Sci.",
  volume		= "32",
  number		= "1",
  pages			= "3",
  year			= "2022"
}

\end{document}